\documentclass[12pt]{amsart}

\usepackage{a4wide,amsfonts,graphicx,amssymb,stmaryrd,amscd,amsmath,latexsym,amsbsy,xypic,amsthm,amsbsy,color,multicol}
\xyoption{all}
\usepackage{marginnote,mathtools,epsfig,enumitem,dsfont,stmaryrd,array,calc,rotating,bm,bbm,mathrsfs} 
 
\usepackage{tikz}
\usetikzlibrary{decorations.markings}
\usetikzlibrary{calc} 
\usetikzlibrary{arrows} 
\usetikzlibrary{patterns} 
\usetikzlibrary{decorations.pathreplacing} 

\numberwithin{equation}{section}
\theoremstyle{plain}
 
\newtheorem{thm}{Theorem}[section]

\newtheorem{pro}[thm]{Proposition}
\newtheorem{defi}[thm]{Definition}
\newtheorem{lem}[thm]{Lemma}
\newtheorem{cor}[thm]{Corollary}

\theoremstyle{remark}
\newtheorem{rema}[thm]{Remark}

\newcommand{\Z}{\mathbb{Z}}

\newcommand{\C}{\mathbb{C}}

\newcommand{\ii}{\mathrm{i}}
\newcommand{\ket}[1]{\left|#1\right\rangle}      
\newcommand{\bra}[1]{\left\langle #1\right|}     

\newcommand{\PROD}[3]{\mathop{\overrightarrow\prod}\limits_{#1 \le #2 \le #3 }}

\newcommand{\hypref}[2]{\ifx\href\asklfhas #2\else\href{#1}{#2}\fi}
\newcommand{\Secref}[1]{Section~\ref{#1}}

\newcommand{\Appref}[1]{Appendix~\ref{#1}}

\renewcommand{\eqref}[1]{(\ref{#1})}

\def\[{\begin{equation}}
\def\]{\end{equation}}
\def\<{\begin{eqnarray}}
\def\>{\end{eqnarray}}

\title[]{On the elliptic $\mbox{\larger$\mathfrak{gl}_2$}$ solid-on-solid model: \\ Functional relations and determinants}
\author{W. Galleas}

\address{Institut f\"ur Theoretische Physik, Eidgen\"ossische Technische Hochschule Z\"urich, Wolfgang-Pauli-Strasse 27, 8093 Z\"urich, Switzerland}

\address{II. Institut f\"ur Theoretische Physik, Universit\"at Hamburg, Luruper Chaussee 149, 22761 Hamburg, Germany.}

\email{galleasw@phys.ethz.ch} 

\thanks{The work of W.G. was partially supported by the German Science Foundation (DFG) under the Collaborative Research Center (SFB) 676: Particles, Strings and the Early Universe.}

\subjclass[2000]{}
\begin{document}

\vspace{2em}
\begin{abstract}
In this work we study an elliptic solid-on-solid model with domain-wall boundaries having the 
elliptic quantum group $\mathcal{E}_{p, \gamma} [\widehat{\mathfrak{gl}_2}]$ as its underlying symmetry algebra. We elaborate on 
results previously presented in \cite{Galleas_pre_2016} and extend our analysis to include continuous families of single determinantal 
representations for the model's partition function. Interestingly, our families of representations are parameterized by two continuous
complex variables which can be arbitrarily chosen without affecting the partition function. 
\end{abstract}

\maketitle
\vspace{-8cm}
\vspace{10cm}

\setcounter{tocdepth}{2}
\tableofcontents

\section{Introduction} \label{sec:INTRO}

Integrable systems have a long history of introducing new concepts and several developments in physics and mathematics can be credited to
their study. Among the list of developments triggered by the study of integrable systems \emph{Quantum Groups} have a special place. 
The formulation of Quantum Groups is one of the main achievements of the \emph{Quantum Inverse Scattering Method} \cite{Sk_Faddeev_1979, Takh_Faddeev_1979}
and its origin is intimately related to the problem of finding solutions of the Yang-Baxter equation. Such solutions are usually 
refereed to as $\mathcal{R}$-matrices and they are central objects within the theory of \emph{quantum integrable systems}. In addition to that,
exactly solvable vertex models of Statistical Mechanics also allow for a formulation having their statistical weights encoded in a $\mathcal{R}$-matrix.

The six-vertex model has played an important role in this development and the constructions nowadays known as \emph{Yangian} \cite{Drinfeld_1985} 
and \emph{$q$-deformation} \cite{Kulish_Reshetikhin_1983} firstly appeared associated to the six-vertex model $\mathcal{R}$-matrix  \cite{Lieb_1967, Sutherland_1967}.
In a  more general setting, for each simple finite-dimensional Lie algebra $\mathfrak{g}$, the Hopf algebra Yangian $\mathcal{Y} [\mathfrak{g}]$ produces 
a rational solution of the Yang-Baxter equation while the $q$-deformation $\mathcal{U}_q [\widehat{\mathfrak{g}}]$ yields a trigonometric one.
However, $\mathcal{R}$-matrices with rational and/or trigonometric entries do not exhaust all possible solutions of the Yang-Baxter equation;
and a distinguished role is played by elliptic solutions. For instance, the eight-vertex model is a corner stone of the theory
of exactly solvable models of Statistical Mechanics and its statistical weights are parameterized by elliptic functions \cite{Baxter_1971}.

On the other hand, among the list of solvable two-dimensional lattice models we also have the so called solid-on-solid models. 
They are closely related to the eight-vertex model and also play a prominent role in Statistical Mechanics. The statistical weights 
of solid-on-solid models are also parameterized by elliptic functions and the algebraic structure ensuring integrability of such models
are nowadays known as \emph{Elliptic Quantum Groups}.

\subsection{Solid-on-solid models}
The introduction of solid-on-solid models, or \textsc{sos} models for short, is intimately related to the study of Baxter's eight-vertex model 
\cite{Baxter_1973}. The former was originally put forward as a tool for describing eigenvectors of the symmetric eight-vertex model
but it has gained life on its own. In the literature they are also refereed to as \emph{interaction-round-a-face} models and their
collective interactions are characterized by variables assigned to the lattice sites instead of edges; the latter being the usual description
employed in vertex models. 

The dimensionality of a lattice model also plays a fundamental role for establishing its integrability in the sense
of Baxter \cite{Baxter_book}; and here we restrict our discussion to models defined on a two-dimensional lattice. 
In addition to that, several equivalences between two-dimensional lattice systems of Statistical Mechanics have 
been established over the years and it is worth remarking that the relation between \textsc{sos} and vertex models configurations was 
firstly noticed by Lenard \cite[Note added in proof]{Lieb_1967}. Further relations have been unveiled over the years and, for instance, 
one also finds that the \textsc{sos} model associated to Baxter's eight-vertex model consists of an Ising-type model with 
four-spin interactions \cite{Baxter_book}. For a detailed discussion on the formulation of 
\textsc{sos} models we refer the reader to \cite{Baxter_book, integrable_book} and references therein.

As far as the eight-vertex model is concerned, the algebraic structure underlying Baxter's elliptic uniformization consists of 
Sklyanin algebra \cite{Sklyanin_1982, Sklyanin_1983}. On the other hand, despite the close connection between the elliptic \textsc{sos}
model considered in the present paper and Baxter's eight-vertex model, the algebraic structure underlying \textsc{sos} models
are not captured by Sklyanin algebras. For \textsc{sos} models one needs to invoke the concept of Elliptic Quantum Groups which was
put forward by Felder in \cite{Felder_1994, Felder_1995}.

\subsection{Elliptic quantum groups} 

The elliptic nature of Baxter's eight-vertex model arises from the requirement that the model's statistical weights satisfy
the Yang-Baxter equation \cite{Baxter_1971, Baxter_1972}. On the other hand, there exists a remarkable relation between the
symmetric eight-vertex model and a particular \textsc{sos} model \cite{Baxter_1973}. Baxter's \emph{vertex-face transformation}
precises this relation and the elliptic nature of the eight-vertex model is consequently transported to the dual \textsc{sos} model. 
In fact, the vertex-face transformation of Baxter's eight-vertex model endows the resulting \textsc{sos} model with an additional continuous
parameter which will be later on refereed to as \emph{dynamical parameter}. 
As far as integrability in the sense of Baxter is concerned, the Yang-Baxter equation for \textsc{sos} models takes the form of the so called \emph{hexagon identity} \cite{integrable_book}. 
The role played by the hexagon identity for \textsc{sos} models is exactly the same as the one played by the standard Yang-Baxter
equation for vertex models: it ensures the model's transfer matrix forms a commuting family. 

The algebraic structures underlying integrable vertex models mostly consist of Drinfel'd-Jimbo quantum enveloping algebras
\cite{Drinfeld_1985, Drinfeld_1987, Jimbo_1985, Jimbo_1986a}, Yangians \cite{Drinfeld_1985, Drinfeld_1987} and Sklyanin algebras 
\cite{Sklyanin_1982, Sklyanin_1983}. As for \textsc{sos} models an analogous structure was only unveiled in \cite{Felder_1994, Felder_1995}
by Felder and it received the name Elliptic Quantum Groups. In \cite{Felder_1994, Felder_1995} Felder also showed that statistical weights
satisfying the hexagon identity are encoded in solutions of a \emph{dynamical} version of the Yang-Baxter equation. This dynamical equation
was proposed as a quantization of a modified classical Yang-Baxter equation arising as the compatibility condition for the 
Knizhnik-Zamolodchikov-Bernard equation \cite{Knizhnik_1984, Bernard_1988a, Bernard_1988b}.
It is worth remarking here that the dynamical Yang-Baxter equation; however, without spectral parameter, had previously appeared in the work 
of Gervais and Neveu on the Liouville string field theory \cite{Gervais_1984}. 
Elliptic Quantum Groups can be defined for any simple Lie algebra $\mathfrak{g}$ and they are usually denoted by 
$\mathcal{E}_{p, \gamma}[\mathfrak{g}]$. They provide algebraic foundations for solutions of the dynamical Yang-Baxter equation 
and in the present paper we shall restrict our discussion to the case $\mathfrak{g} = \mathfrak{gl}_2$.

\subsection{Domain-wall boundaries} Three main ingredients are required in order to define a vertex or \textsc{sos} model: graphs on a lattice,
statistical weights for graphs configurations and boundary conditions. Integrability in the sense of Baxter imposes restriction on all of them
and, in particular, inappropriate choices of boundary conditions can render the model trivial or break its bulk integrability.
The Yang-Baxter equation and its dynamical version govern the model's bulk integrability, while integrability preserving boundary
conditions are usually singled out from additional set of constraints. See for instance \cite{deVega_1984} and \cite{Sklyanin_1988}.
However, there still exist certain types of boundary conditions rendering the model exactly solvable which are not characterized by 
extra algebraic equations. This is precisely the case of domain-wall boundary conditions introduced by Korepin in order to study scalar 
products of Bethe vectors for the six-vertex model \cite{Korepin_1982}. 
Such boundary conditions render a well defined system of Statistical Mechanics and, interestingly, the partition function of 
the six-vertex model with domain-wall boundaries has been exactly computed in a closed form. It can be expressed as determinants 
\cite{Izergin_1987, Colomo_Pronko_2008, Galleas_2016} and multiple contour integrals \cite{deGier_Galleas_2011, Galleas_2012, Galleas_2013};
in contrast to the case with toroidal boundary conditions whose partition function evaluation still depends on the resolution of Bethe ansatz
equations \cite{Lieb_1967}.

Boundary conditions of domain-wall type have also been considered for \textsc{sos} models, see for instance 
\cite{Razumov_2009a, Razumov_2009b, Pakuliak_2008, Rosengren_2009, Galleas_2012, Galleas_2013}. In particular, the model's partition function
was shown in \cite{Razumov_2009a} to satisfy a rather simple functional equation when the anisotropy parameter satisfies a root-of-unity
condition. For generic values of the anisotropy, we have put forward different kinds of functional relations having their origins in the
dynamical Yang-Baxter algebra \cite{Galleas_2012, Galleas_2013}. The resolution of those functional equations resulted in a multiple contour
integral representation for the model's partition function for generic values of the anisotropy parameter. 
Furthermore, the authors of \cite{Pakuliak_2008} have shown that the partition function of the elliptic \textsc{sos} model with domain-wall
boundaries consists of an \emph{universal elliptic weight function}. The latter relation has also led to a sum over permutations representation
for the model's partition function. As far as determinantal representations are concerned, in the work \cite{Rosengren_2009} Rosengren has 
presented a sum of Frobenius type determinants which seems to generalize Izergin's representation for the six-vertex model \cite{Izergin_1987}.
More recently, Rosengren has also shown that the aforementioned partition function can be written as the sum of two \emph{pfaffians} when 
the model's anisotropy parameter is fixed to a particular root-of-unity value \cite{Rosengren_2016}.

It is worth remarking that \textsc{sos} models with domain-walls and one reflecting end have also been studied in the literature.
Interestingly, in that case the model's dynamical structure does not offer obstacles for writing down the model's partition
function as a single determinant \cite{Filali_2010, Filali_2011}. Functional equations describing the latter partition function 
are also available and they have been obtained in \cite{Lamers_2015} generalizing the results of  \cite{Galleas_Lamers_2014} for the
trigonometric (non-dynamical) model. The resolution of the aforementioned functional equations has produced multiple contour integral
representations for the corresponding partition functions.

\subsection{Applications} 
Vertex models were proposed by Pauling in 1935 in order to explain the ice residual entropy \cite{Pauling_1935} and several other
applications have emerged over the years. In their turn, solid-on-solid models enjoy a similar status and some applications involve 
the description of physical systems whilst others have a purely mathematical scope. As far as physical applications are concerned, 
the case with periodic boundary conditions fills most of the literature up to the present days. For instance, the case known as \emph{restricted} \textsc{sos} 
model \cite{Andrews_Baxter_Forrester_1984, Forrester_Baxter_1985} was shown to realize certain unitary minimal models 
in the continuum scaling limit \cite{Huse_1984}. Also, the restricted \textsc{sos} model is known to describe non-local statistics of
height clusters and percolation hull exponents \cite{Saleur_Duplantier_1987}. 

Physical applications of the elliptic \textsc{sos} model with domain-wall boundary conditions are still rather limited and this is possibly due to
the lack of appropriate representations for the model's partition function. However, this model exhibits a very rich structure and 
interesting mathematical applications have been found over the past years. The partition function of the elliptic \textsc{sos} model with
domain-wall boundaries has been studied through several approaches, see for instance \cite{Pakuliak_2008, Rosengren_2009, Galleas_2013}, 
and remarkable properties have been reported when the model degenerates into the so-called three-color model. The latter is obtained by 
fixing the model's anisotropy parameter as a particular root-of-unity value coinciding with Kuperberg's specialization used for the
six-vertex model \cite{Kuperberg_1996}. The trigonometric six-vertex model under Kuperberg's specialization is known to enumerate 
\emph{Alternating Sign Matrices} (\textsc{ASM}) and a similar analysis for the three-color model has been performed in \cite{Rosengren_2009}. 
In that case Rosengren has found that, besides the enumeration of \textsc{ASM}, the three-color model also counts the number of
\emph{Cyclically Symmetric Plane Partitions} (\textsc{CSPP}). 

In addition to applications in enumerative combinatorics the three-color model with domain-wall boundaries has also found applications in 
the theory of \emph{special functions}. For instance, the work \cite{Rosengren_2011} demonstrates that the aforementioned model gives rise 
to a family of two-variables orthogonal polynomials satisfying special recursion relations. Using these results Rosengren has conjectured
an explicit formula for the model's free-energy in the thermodynamical limit  \cite{Rosengren_2011}. Other degeneration of the elliptic
\textsc{sos} model also exhibit special properties. The so-called \emph{supersymmetric point} is one of them and, at this point, the model's
partition function gives rise to certain symmetric polynomials whose properties have been studied in the series of works 
\cite{Rosengren_2013a, Rosengren_2013b, Rosengren_2014, Rosengren_2015}.

\subsection{Algebraic-functional approach}
Functional methods are at the core of the modern theory of integrable systems and the functional equation derived in \cite{Galleas_2013}
characterizing the partition function of the elliptic \textsc{sos} model with domain-wall boundaries has a singular origin: the dynamical
Yang-Baxter algebra. In particular, the equation of \cite{Galleas_2013} exhibits similarities with Knizhnik-Zamolodchikov equations 
and this feature has been already discussed in \cite{Galleas_proc}. Multiple contour integrals are known to accommodate solutions of
Knizhnik-Zamolodchikov equations \cite{Tarasov_1997, Varchenko_book, Etingof_book} and this feature has motivated the search for similar
representations for the solution of our functional equation.

As far as the origin of our equation is concerned, the possibility of exploiting the Yang-Baxter algebra and its dynamical version in such a
manner did not appear for the first time in \cite{Galleas_2013}. It was firstly put forward in \cite{Galleas_2008} for spectral problems and 
in \cite{Galleas_2010} for partition functions with special boundary conditions. Scalar products of Bethe vectors have been tackled 
through this method as well, see \cite{Galleas_SCP, Galleas_openSCP}, resulting in compact representations for \emph{off-shell} scalar products.
Furthermore, in \cite{Galleas_2011, Galleas_proc} we have also pointed out the existence of families of \emph{Partial Differential Equations}
underlying such functional relations. The mechanism leading to such differential equations have been further developed in the works
\cite{Galleas_2015, Galleas_Lamers_2014, Galleas_Lamers_2015} which includes the analysis of the toroidal six-vertex model spectrum
and \emph{on-shell} scalar products of Bethe vectors. 
On the other hand, an alternative mechanism leading to a system of first-order partial differential equations has been recently described
in \cite{Galleas_2016}. Although the results of \cite{Galleas_2016} are restricted to the rational six-vertex model with domain-wall
boundaries, whose functional equation follows from a particular limit of the one presented in \cite{Galleas_2013}, the generalization
of \cite{Galleas_2016} to the trigonometric case is straightforward.

The system of partial differential equations found in \cite{Galleas_2016} exhibits uncanny similarities with classical Knizhnik-Zamolodchikov
equations and its resolution has produced a discrete family of single determinant representations for the six-vertex model's partition function. 
The representations found in \cite{Galleas_2016} do not seem to be reminiscent from Izergin's determinant \cite{Izergin_1987} and,
by way of contrast, they offer trivial access to the model's partial homogeneous limit.

The algebraic structures underlying quantum integrable systems are not limited to the Yang-Baxter algebra and its dynamical version.
For instance, integrable systems with open boundary conditions also obey the so called \emph{reflection algebra}. The latter algebraic 
structure can also be exploited along the lines of the algebraic-functional framework and this possibility was demonstrated in 
\cite{Galleas_Lamers_2014}. Moreover, this approach is also feasible using a dynamical version of the reflection algebra 
as shown by Lamers in \cite{Lamers_2015}.

\subsection{The goals of this paper}
Integrable systems usually share common structures and one relevant question which has been debated over the past years is if the
partition function of the elliptic \textsc{sos} model with domain-wall boundaries can be accommodated by a single determinant formula.
As a matter of fact, the results presented in the literature so far suggested that the answer for this question was negative. 
For instance, the representation found by Rosengren in \cite{Rosengren_2009} is given in terms of a sum of determinants which reduces
to Izergin's single determinant formula in the six-vertex model limit. This was actually a good indication that Rosengren's formula is
indeed the natural extension of \cite{Izergin_1987} and, therefore, single determinant representations should not exist. 

However, the family of representations recently found in \cite{Galleas_2016} do not share any similarity with Izergin's formula,
although they are also expressed as a single determinant. In addition to that, the method employed in \cite{Galleas_2016} does not 
exhibit any overlap with the recursive approach of Korepin \cite{Korepin_1982} which ultimately leads to the results of \cite{Izergin_1987}.

Given the scenario above described it is natural to ask if the approach proposed in \cite{Galleas_2016} can be extended to 
the elliptic \textsc{sos} model. Moreover, if that is indeed possible, will the solution be given (formally) by a single determinant?
The answer for this question is positive and the derivation of such representations along the lines of the Algebraic-Functional 
framework is what we intend to discuss in the present work.

\vspace{1em}
\paragraph{{\bf Outline.}}
We have organized this paper placing the spotlight on the approach leading to our main result. The formulation of the model under
consideration has been already extensively discussed in the literature and we shall not go into those details.
The mathematical definitions and conventions employed throughout this work can be found in \Secref{sec:DEF}. The main ideas of the 
Algebraic-Functional (AF) method are presented in \Secref{sec:AFM}, although some technical aspects of the AF method are omitted since they
have been discussed in previous works. On the other hand, generalizations required to obtain the announced determinantal representations
will be discussed in some detail. The functional relation used to obtain our determinantal representations is derived in \Secref{sec:AD}
and its resolution is discussed in \Secref{sec:PF}. We leave \Secref{sec:6V} for the analysis of the six-vertex model limit
and the proofs of our main theorems are given in \Appref{app:proofT1} and \Appref{app:proofAD}. 
For the sake of clarity, we gather explicit formulae required for building our determinantal representations in 
\Appref{app:FUN}.

\section{Definitions and conventions} \label{sec:DEF}

This work is concerned with a particular two-dimensional lattice model of Statistical Mechanics, alias elliptic \textsc{sos} model, whose statistical weights
are intimately related to the elliptic quantum group $\mathcal{E}_{p, \gamma}[\widehat{\mathfrak{gl}_2}]$. This model has been 
extensively discussed in the literature and, in this way, we shall restrict ourselves to presenting only the mathematical definitions required
throughout this work. The model definition includes statistical weights for plaquette's configurations and boundary conditions characterizing
the relevant partition function. As far as those ingredients are concerned, we shall mostly use the conventions of \cite{Galleas_2012, Galleas_2013}.
Statistical weights of integrable \textsc{sos} models can be encoded in solutions of the dynamical Yang-Baxter equation \cite{Felder_1994, Felder_1995}
and this will be the starting point of our discussion.

\subsection{Dynamical Yang-Baxter equation} \label{sec:DYBE}
Let $\mathcal{V}$ be a complex vector space and $\mathfrak{gl}(\mathcal{V})$
the associated Lie algebra. Here we shall use $\mathfrak{h}$ to denote the Cartan subalgebra of $\mathfrak{gl}(\mathcal{V})$ which 
is also endowed with a symmetric bilinear form $(\cdot, \cdot) \colon \mathfrak{h} \times \mathfrak{h} \to \C$.
Moreover, let $\mathcal{W}$ be a finite dimensional representation of $\mathcal{V}$ and $\pi \colon \mathfrak{gl}(\mathcal{V}) \to \mathfrak{gl}(\mathcal{W})$
be the representation map. In addition to that we also define the weight module $\mathcal{W}[\lambda]$ as
\[ \label{wmd}
\mathcal{W}[\lambda ] \coloneqq \{ w \in \mathcal{W} \mid \pi (h) w = (h,\lambda) w \quad \forall h \in \mathfrak{h} \} \; .
\]
In \eqref{wmd} $\lambda \in \mathfrak{h}$ is refereed to as weight of $\mathcal{W}$ if $\mathcal{W}[\lambda ]$ is non-empty; and 
we write $\Lambda(\mathcal{W}) \subset \mathfrak{h}$ for the set of weights of $\mathcal{W}$. Then, as representation of
$\mathfrak{gl}(\mathcal{V})$, $\mathcal{W}$ admits the weight decomposition
\[
\mathcal{W} = \bigoplus_{\lambda \in \Lambda(\mathcal{W})} \mathcal{W}[\lambda] \; .
\]

Next we suppose $\mathcal{Q}(\zeta) \colon \mathcal{W}^{\otimes n} \to \mathcal{W}^{\otimes n}$ is a linear operator depending meromorphically
on $\zeta \in \mathfrak{h}$. We refer to $\mathcal{Q}(\zeta)$ as dynamical operator when it acts on $\mathcal{W}^{\otimes n}$ 
according to the direct sum decomposition
\[
\mathcal{W}^{\otimes n} = \bigoplus_{\lambda \in \Lambda(\mathcal{W})} \mathcal{W}^{\otimes (i-1)} \otimes  \mathcal{W}[\lambda] \otimes \mathcal{W}^{\otimes (n-i)} \; . 
\]
More precisely, for $\gamma \in \C$ and $1\leq i \leq n$, we write $\mathcal{Q}(\zeta + \gamma h_i)$ for the operator acting as 
$\mathcal{Q}(\zeta + \gamma \lambda)$ on the projected subspace $\mathcal{W}^{\otimes (i-1)} \otimes  \mathcal{W}[\lambda] \otimes \mathcal{W}^{\otimes (n-i)}$
of $\mathcal{W}^{\otimes n}$.

Now let $\mathcal{R}(x; \zeta) \in \mbox{End}(\mathcal{W}^{\otimes 2})$ be a dynamical linear operator whose
dependence on $x \in \C$ and $\zeta \in \mathfrak{h}$ is meromorphic. Then, using the standard tensor leg notation, the dynamical Yang-Baxter
equation is a relation in $\mbox{End}(\mathcal{W}^{\otimes 3})$ reading
\< \label{dyb}
\mathcal{R}_{12} (x_1 - x_2; \tau - \gamma h_3) \mathcal{R}_{13} (x_1 - x_3; \tau) \mathcal{R}_{23} (x_2 - x_3; \tau - \gamma h_1) = \nonumber \\
\mathcal{R}_{23} (x_2 - x_3; \tau ) \mathcal{R}_{13} (x_1 - x_3; \tau - \gamma h_2)  \mathcal{R}_{12} (x_1 - x_2; \tau) \; ,
\>
for $x_i, \tau, \gamma \in \C$.

\subsection{The $\mathcal{E}_{p, \gamma}[ \widehat{\mathfrak{gl}_2} ]$ dynamical $\mathcal{R}$-matrix} \label{sec:EQG}
The specialization of \eqref{dyb} to the $\mathfrak{gl}_2$ case is obtained by setting $\mathcal{V} = \C v_1 \oplus \C v_2$ with basis vectors $v_i \in \C^2$. Next we identify 
$\mathfrak{gl}(\mathcal{V})$ with a matrix algebra via the ordered basis $\{ v_1 , v_2 \} \in \mathcal{V}$ in such a way that  
$\mathfrak{gl}(\mathcal{V} ) \simeq \mathfrak{gl}_2$. 
The Cartan subalgebra is then $\mathfrak{h} \coloneqq \C E_{11} \oplus \C E_{22}$ with matrix units $E_{ij} \in \mathfrak{gl}(\mathcal{V} )$ 
defined through the action $E_{i,j} (v_k) \coloneqq \delta_{j,k} v_i$ for $i,j, k \in \{ 1 , 2  \}$. Furthermore, the symmetric bilinear form 
$(\cdot, \cdot) \colon \mathfrak{h} \times \mathfrak{h} \to \C$ is then given by
\[
(E_{ii} , E_{jj} ) = \begin{cases}
1 \qquad \qquad \mbox{if} \; i=j=1 \cr
-1 \; \quad \qquad \mbox{if} \; i=j=2 \cr
0 \qquad \qquad \mbox{otherwise} \end{cases} .
\]
For later convenience we also fix the elliptic nome $0 < p < 1$ and introduce the short hand notation
\[ \label{theta}
[x] \coloneqq \frac{1}{2} \sum_{n = -\infty}^{+ \infty} (-1)^{n-\frac{1}{2}} p^{(n+\frac{1}{2})^2} e^{-(2n+1) x}
\]
for $x \in \C$. 
\begin{rema} 
We stress here that $[x]$ corresponds to the Jacobi theta-function $\Theta_1 (\ii x, \nu)$ with $p = e^{\ii \pi \nu}$
according to the conventions of \cite{Whittaker_Watson_book}.
\end{rema}
Given the above definitions we fix the parameter $\gamma \in \C$ in such a way that the solution of \eqref{dyb} associated to
the elliptic quantum group $\mathcal{E}_{p, \gamma}[ \widehat{\mathfrak{gl}_2} ]$ explicitly reads
\< \label{rmat}
\mathcal{R}(x;\tau) = \left(\begin{matrix}
a_{+}(x, \tau) & 0 & 0 & 0 \cr
0 & b_{+}(x, \tau)  & c_{+}(x, \tau)  & 0 \cr
0 & c_{-}(x, \tau)  & b_{-}(x, \tau)  & 0 \cr
0 & 0  & 0 & a_{-}(x; \tau) 
\end{matrix} \right)
\>
with non-null entries defined by
\< \label{bw}
a_{\pm} (x, \tau) &\coloneqq& [x + \gamma] \nonumber \\
b_{\pm} (x, \tau) &\coloneqq&  [\tau \mp \gamma] [x] [\tau]^{-1} \nonumber \\
c_{\pm} (x, \tau) &\coloneqq& [\tau \mp x] [\gamma]  [\tau]^{-1}  \; .
\>
The construction of \eqref{rmat} relies on the representation theory of the elliptic quantum group 
$\mathcal{E}_{p, \gamma}[ \widehat{\mathfrak{gl}_2} ]$. The latter has been described in \cite{Felder_Varchenko_1996}.

\subsection{Modules over $\mathcal{E}_{p, \gamma}[ \widehat{\mathfrak{gl}_2} ]$} \label{sec:MEQG}
Following \cite{Felder_1994, Felder_1995} we refer to the algebra associated with the dynamical $\mathcal{R}$-matrix \eqref{rmat} as 
elliptic quantum group $\mathcal{E}_{p, \gamma}[ \widehat{\mathfrak{gl}_2} ]$. We shall also employ $\mathscr{A}(\mathcal{R})$ to denote this algebra and it is generated by meromorphic functions
on $\mathfrak{h}$ and non-commutative matrix elements of $\mathcal{L}(x, \tau) \in \mbox{End}(\mathcal{V})$ subjected to the following relation
\< \label{dyba}
\mathcal{R}_{12}(x_1 - x_2; \tau - \gamma h) \mathcal{L}_1 (x_1, \tau) \mathcal{L}_2 (x_2, \tau - \gamma h_1) = \nonumber \\
\mathcal{L}_2 (x_2, \tau ) \mathcal{L}_1 (x_1, \tau - \gamma h_2) \mathcal{R}_{12}(x_1 - x_2; \tau ) \; .
\>
The generator $h$ in \eqref{dyba} is an element of the $\mathfrak{gl}_2$ Cartan subalgebra $\mathfrak{h}$ and $\mathcal{L}$ is regarded as a dynamical
operator. Here we are mainly interested on representations of $\mathscr{A}(\mathcal{R})$ consisting of a diagonalizable
$\mathfrak{h}$-module $\widetilde{\mathcal{W}}$ together with a meromorphic function $\mathcal{L}(x, \tau)$ on $\C \times \mathfrak{h}$ with values in
$\mbox{End}_{\mathfrak{h}} (\mathcal{W} \otimes \widetilde{\mathcal{W}})$ such that \eqref{dyba} is fulfilled on $\mathcal{W} \otimes \mathcal{W} \otimes \widetilde{\mathcal{W}}$. 
Therefore, a representation of $\mathscr{A}(\mathcal{R})$ is defined by the pair $(\widetilde{\mathcal{W}} , \mathcal{L} )$; and the one consisting of 
$\widetilde{\mathcal{W}} = \mathcal{W}$ and $\mathcal{L}(x, \tau) = \mathcal{R}(x-\mu; \tau)$, with $\mathcal{R}$-matrix given by 
\eqref{rmat}, is a module over $\mathcal{E}_{p, \gamma}[ \widehat{\mathfrak{gl}_2} ]$ usually refereed to as fundamental representation with 
evaluation point $\mu$. For short we also refer to modules over $\mathcal{E}_{p, \gamma}[ \widehat{\mathfrak{gl}_2} ]$ as 
$\mathcal{E}$-modules and remark that more general $\mathcal{E}$-modules have  been constructed in \cite{Felder_Varchenko_1996}. 

\subsection{Dynamical monodromy matrix} \label{sec:DMM}
One of the building blocks of the Quantum Inverse Scattering Method \cite{Sk_Faddeev_1979, Takh_Faddeev_1979}
is the so called \emph{monodromy matrix} and the following theorem paves the way for defining its dynamical version.

\begin{thm}[Felder] \label{Felder}
Let $(\widetilde{\mathcal{W}}', \mathcal{L}' )$ and $(\widetilde{\mathcal{W}}'', \mathcal{L}'' )$ be $\mathcal{E}$-modules. Then 
$(\widetilde{\mathcal{W}} , \mathcal{L} )$ is also an $\mathcal{E}$-module with $\widetilde{\mathcal{W}} = \widetilde{\mathcal{W}}' \otimes \widetilde{\mathcal{W}}''$
enjoying $\mathfrak{h}$-module structure $h(\tilde{w}' \otimes \tilde{w}'') = (h \tilde{w}') \otimes \tilde{w}'' + \tilde{w}' \otimes (h \tilde{w}'')$
for $h \in \mathfrak{h}$, $\tilde{w}' \in \widetilde{\mathcal{W}}'$, $\tilde{w}'' \in \widetilde{\mathcal{W}}''$, and $\mathcal{L} = \mathcal{L}'(x-x', \tau - \gamma h_2) \mathcal{L}'' (x-x'', \tau)$. 
\end{thm}

Next we would like to build $\mathcal{E}$-modules, or representations of $\mathscr{A} (\mathcal{R})$, on 
$\widetilde{\mathcal{W}} = \mathcal{W}^{\otimes n}$. This can be achieved by iterating Theorem \ref{Felder}.

\begin{defi}
Let $L$ be a positive integer and fix parameters $\mu_i$ for $1 \leq i \leq L$. Also, let $\mathcal{W}_0 = \mathcal{W}_i = \mathcal{W} \simeq \C^2$
and define the dynamical monodromy matrix $\mathcal{T}_0 \in \mbox{End}(\mathcal{W}_0 \otimes \mathcal{W}^{\otimes L})$ as
\[ \label{mono}
\mathcal{T}_0 (x, \tau) \coloneqq \PROD{1}{i}{L} \mathcal{R}_{0 i} (x - \mu_i; \tau - \gamma \sum_{k=i+1}^{L} h_k) \; .
\]
The object $\mathcal{R}_{0 i}$ in \eqref{mono} corresponds to the $\mathcal{R}$-matrix \eqref{rmat} embedded 
in $\mbox{End}(\mathcal{W}_0 \otimes \mathcal{W}_i)$. 
Using Theorem \ref{Felder} one can show the pair $(\mathcal{W}^{\otimes L}, \mathcal{T}_0)$ is an $\mathcal{E}$-module.
\end{defi}

In its turn, $\mathscr{A} (\mathcal{R})$ as an algebra contains relations for the non-commutative entries of $\mathcal{L}(x, \tau)$. 
Here, however, we shall focus on a particular representation of $\mathscr{A} (\mathcal{R})$, namely $\mathcal{T}_0 (x, \tau)$ defined in \eqref{mono}. In that case
$\mathcal{W} \simeq \C^2$ and the relations in $\mathscr{A} (\mathcal{R})$ can be conveniently described through the structure
\[ \label{abcd}
\mathcal{T}_0 (x, \tau) \eqqcolon \left( \begin{matrix}
\mathcal{A}(x , \tau) & \mathcal{B}(x , \tau) \cr
\mathcal{C}(x , \tau) & \mathcal{D}(x , \tau) \cr
\end{matrix} \right)
\]
with generators $\mathcal{A}, \mathcal{B}, \mathcal{C}, \mathcal{D} \in \mbox{End}(\mathcal{W}^{\otimes L})$.

\subsection{Highest-weight modules} \label{sec:HWM}
One important feature of $\mathscr{A} (\mathcal{R})$ is that the notion of highest-weight module is well defined. 
The latter was introduced in \cite{Felder_Varchenko_1996} for elliptic quantum groups.
The starting point for building highest-weight modules is the concept of \emph{singular vectors} in an $\mathcal{E}$-module 
$(\mathcal{W}^{\otimes L} , \mathcal{T}_0 )$. They consist of non-zero elements $w_0 \in \mathcal{W}^{\otimes L}$
such that $\mathcal{C}(x , \tau) w_0 = 0$ for all $x, \tau \in \C$. Similarly, we let $\bar{\mathcal{W}}^{\otimes L}$ be the vector
space dual to $\mathcal{W}^{\otimes L}$ and call \emph{dual singular vector} the non-zero elements $\bar{w}_0 \in \bar{\mathcal{W}}^{\otimes L}$ such that
$\bar{w}_0 \; \mathcal{C}(x , \tau) = 0$ for all $x, \tau \in \C$.

The representation $\mathcal{W}^{\otimes L}$ is a diagonalizable $\mathfrak{h}$-module and we say an element $w \in \mathcal{W}^{\otimes L}$
has $\mathfrak{h}$-weight $\mu$ if $h w = \mu w$ for $h \in \mathfrak{h}$. Additionally, we assign the weight $(\mu, \lambda_{\mathcal{A}}(x, \tau) , \lambda_{\mathcal{B}}(x, \tau) )$
to an element $w \in \mathcal{W}^{\otimes L}$ if $\mathcal{A}(x, \tau) w = \lambda_{\mathcal{A}}(x, \tau) w$, $\mathcal{D}(x, \tau) w = \lambda_{\mathcal{D}}(x, \tau) w$  
and $w$ has $\mathfrak{h}$-weight $\mu$. Also, we restrict our attention to the cases where $\lambda_{\mathcal{A}, \mathcal{D}}$ are 
non-vanishing meromorphic functions on $x$ and $\tau$.

Now we have gathered the ingredients required to define a highest-weight module. They are formed by singular vectors $w \in \mathcal{W}^{\otimes L}$
having weight $(\mu, \lambda_{\mathcal{A}}(x, \tau) , \lambda_{\mathcal{B}}(x, \tau) )$. Dual highest-weight modules are defined similarly
by considering dual singular vectors instead of singular vectors. Hence, we say $\bar{w} \in \bar{\mathcal{W}}^{\otimes L}$ is a dual
highest-weight vector with weight $(\bar{\mu}, \bar{\lambda}_{\mathcal{A}}(x, \tau) , \bar{\lambda}_{\mathcal{B}}(x, \tau) )$ if it is
a dual singular vector satisfying the conditions $\bar{w} h = \bar{\mu} \bar{w}$ for $h \in \mathfrak{h}$, $\bar{w} \; \mathcal{A}(x, \tau)  = \bar{\lambda}_{\mathcal{A}}(x, \tau) \; \bar{w}$ and 
$\bar{w} \; \mathcal{D}(x, \tau)  = \bar{\lambda}_{\mathcal{D}}(x, \tau) \; \bar{w}$.

Next we specialize our discussion to the $\mathfrak{gl}_2$ case and consider $H \coloneqq E_{11} - E_{22} \in \mathfrak{h}$.
From \eqref{mono} and \eqref{rmat} we can take the following vectors as highest- and dual highest-weight vectors respectively:
\< \label{zero}
w_0 = \ket{0} \coloneqq \left( \begin{matrix} 1 \cr 0 \end{matrix} \right)^{\otimes L} \qquad \mbox{and} \qquad \bar{w}_0 = \bra{\bar{0}} \coloneqq \left( \begin{matrix} 0 &  1 \end{matrix} \right)^{\otimes L}  \; .
\>
In this way $w_0$ is a highest-weight vector with weight $( L, \Lambda_{\mathcal{A}}, \Lambda_{\mathcal{D}})$ where
\< \label{lambdas}
\Lambda_{\mathcal{A}}(x, \tau) &\coloneqq& \prod_{j=1}^L [x - \mu_j + \gamma] \nonumber \\
\Lambda_{\mathcal{D}}(x, \tau) &\coloneqq&  \frac{[\tau + \gamma]}{[\tau + (1-L)\gamma]} \prod_{j=1}^L [x - \mu_j] \; .
\>
Analogously, $\bar{w}_0$ is a dual highest-weight vector with weight $( -L, \bar{\Lambda}_{\mathcal{A}}, \bar{\Lambda}_{\mathcal{D}})$
and functions $\bar{\Lambda}_{\mathcal{A}, \mathcal{D}}$ defined as 
\< \label{blambdas}
\bar{\Lambda}_{\mathcal{A}}(x, \tau) &\coloneqq& \frac{[\tau - \gamma]}{[\tau + (L-1)\gamma]} \prod_{j=1}^L [x - \mu_j] \nonumber \\
\bar{\Lambda}_{\mathcal{D}}(x, \tau) &\coloneqq& \prod_{j=1}^L [x - \mu_j + \gamma] \; .
\>

\subsection{Partition function} \label{sec:PF}
The partition function of the $\mathcal{E}_{p, \gamma} [\widehat{\mathfrak{gl}_2}]$ elliptic \textsc{sos} model with domain-wall
boundary conditions admits an operatorial description closely related to that of the six-vertex model \cite{Korepin_1982}.
As shown in \cite{Galleas_2013}, it can be written as
\[ \label{pf}
Z_{\tau} (x_1, x_2 , \dots , x_L) = \bra{\bar{0}} \PROD{1}{j}{L} \mathcal{B}(x_j, \tau + j \gamma) \ket{0} 
\]
with operator $\mathcal{B}$ defined in \eqref{abcd}. Here we fix the parameters $\gamma, \mu_j \in \C$ and omit their dependence in the
LHS of \eqref{pf}.

\section{Algebraic-functional method} \label{sec:AFM}

The algebra $\mathscr{A}(\mathcal{R})$ associated to the $\mathcal{R}$-matrix \eqref{rmat} is an algebra over $\C$ generated by
meromorphic functions $f$ on $h \in \mathfrak{h}$ and elements $\mathcal{A}$, $\mathcal{B}$, $\mathcal{C}$ and $\mathcal{D}$
defined in \eqref{abcd}. We also refer to $\mathscr{A}(\mathcal{R})$ as \emph{dynamical Yang-Baxter algebra} and it 
contains two groups of relations. The first group involves generic functions $f$ and $g$ on the Cartan subalgebra $\mathfrak{h}$ with relations reading
\begin{align} \label{cartan}
f(h) g(h) &= g(h) f(h) \nonumber \\
\mathcal{A} (x, \tau) f(h) &= f(h) \mathcal{A} (x, \tau) & \mathcal{D} (x, \tau) f(h) &= f(h) \mathcal{D} (x, \tau) \nonumber \\
\mathcal{B} (x, \tau) f(h) &= f(h + 2) \mathcal{B} (x, \tau) & \mathcal{C} (x, \tau) f(h) &= f(h-2) \mathcal{C} (x, \tau) \; .
\end{align}
The second group is encoded in \eqref{dyba}.

We also write $\mathscr{A}_2 (\mathcal{R})$ for $\mathscr{A}(\mathcal{R})$ regarded as a matrix algebra with elements in
$\C [x_1^{\pm 1}, x_2^{\pm 1}] \otimes \text{End}(\mathcal{W}^{\otimes L})$. Moreover, we consider 
$\mathscr{M}_n \coloneqq \{ \mathcal{A}(x_n, \tau) , \mathcal{B}(x_n, \tau) , \mathcal{C}(x_n, \tau) , \mathcal{D}(x_n, \tau), f(h) \}$
for meromorphic functions $f$ on $h \in \mathfrak{h}$, in such a way that the repeated use of $\mathscr{A}(\mathcal{R})$
yields relations in $\mathscr{A}_n (\mathcal{R}) \cong \mathscr{A}_{n-1} (\mathcal{R}) \otimes \mathscr{M}_n  / \mathscr{A}_2 (\mathcal{R})$.
Here we refer to $\mathscr{A}_n (\mathcal{R})$ as \emph{dynamical Yang-Baxter algebra of degree $n$}.

The main idea of the AF method is to use $\mathscr{A}_n (\mathcal{R})$ as a source of functional relations characterizing
quantities of interest. This is possible if we are able to exhibit a linear functional 
$\Phi \colon \mathscr{A}_n (\mathcal{R}) \to \C [x_1^{\pm 1} , x_2^{\pm 1} , \dots , x_n^{\pm 1}]$ such that
\begin{align} \label{condO}
\Phi  \left( J \; \mathcal{A}(x_{n}, \tau)  \right) &= \omega_{\mathcal{A}} (x_{n}, \tau) \;  \Phi  \left( J  \right) & \Phi  \left( J \; \mathcal{D}(x_{n}, \tau)  \right) & = \omega_{\mathcal{D}} (x_{n}, \tau) \; \Phi  \left( J  \right) \nonumber \\
\Phi  \left( \mathcal{A}(x_{n}, \tau) \; J  \right) &= \bar{\omega}_{\mathcal{A}} (x_{n}, \tau) \;  \Phi  \left( J  \right) & \Phi  \left( \mathcal{D}(x_{n}, \tau) \; J  \right) &= \bar{\omega}_{\mathcal{D}} (x_{n}, \tau) \;  \Phi  \left( J  \right)
\end{align}
for fixed functions $\omega_{\mathcal{A}, \mathcal{D}}$, $\bar{\omega}_{\mathcal{A}, \mathcal{D}}$ and $J \in \mathscr{A}_{n-1} (\mathcal{R})$.
The construction of $\Phi$ will then depend on the quantity we would like to describe.

\begin{pro}
Taking into account the operatorial description \eqref{pf} we can conveniently realize $\Phi$ as
\[ \label{O}
\Phi \left( J \right) = \bra{\bar{0}} J \ket{0} \qquad J \in \mathscr{A}_n (\mathcal{R}) \; ,
\]
with vectors $\bra{\bar{0}}$ and $\ket{0}$ defined in \eqref{zero}.
\end{pro}
\begin{proof}
Direct use of properties described in \Secref{sec:HWM}.
\end{proof}
\begin{rema} 
Conditions \eqref{condO} are fulfilled with $\omega_{\mathcal{A} , \mathcal{D}} = \Lambda_{\mathcal{A} , \mathcal{D}}$ and 
$\bar{\omega}_{\mathcal{A} , \mathcal{D}} = \bar{\Lambda}_{\mathcal{A} , \mathcal{D}}$ defined in \eqref{lambdas} and \eqref{blambdas}.
\end{rema}

In what follows we shall use the above described procedure to derive two functional relations, namely equations type A and type D,
characterizing the partition function of the elliptic \textsc{sos} model with domain-wall boundaries.

\subsection{Equation type A} \label{sec:typeA}

The algebra $\mathscr{A}_2 (\mathcal{R})$ amounts to $21$ commutation relations: $5$ involving functions on $\mathfrak{h}$ \eqref{cartan}
and $16$ encoded in \eqref{dyba}. In order to obtain a functional relation of \emph{type A}, we can restrict ourselves to the subalgebra
$\mathscr{S}_{\mathcal{A}, \mathcal{B}} \subset \mathscr{A}_2 (\mathcal{R})$ formed by 
\< \label{SAB}
\mathcal{A}(x_1, \tau) \mathcal{A}(x_2, \tau - \gamma) &=& \mathcal{A}(x_2, \tau) \mathcal{A}(x_1, \tau - \gamma) \nonumber \\
\mathcal{B}(x_1, \tau) \mathcal{B}(x_2, \tau + \gamma) &=& \mathcal{B}(x_2, \tau) \mathcal{B}(x_1, \tau + \gamma)  \nonumber \\
\mathcal{A}(x_1, \tau) \mathcal{B}(x_2, \tau - \gamma) &=& \frac{[x_2 - x_1 + \gamma][\tau]}{[x_2 - x_1][\tau + \gamma]} \mathcal{B}(x_2, \tau) \mathcal{A}(x_1, \tau + \gamma) \nonumber \\
&& + \; \frac{[\tau + x_1 - x_2][\gamma]}{[x_1 - x_2][\tau + \gamma]} \mathcal{B}(x_1, \tau) \mathcal{A}(x_2, \tau + \gamma) \nonumber \\
\mathcal{B}(x_1, \tau) \mathcal{A}(x_2, \tau + \gamma) &=& \frac{[x_2 - x_1 + \gamma][\tau]}{[x_2 - x_1][\tau - \gamma]} \mathcal{A}(x_2, \tau) \mathcal{B}(x_1, \tau - \gamma) \nonumber \\
&& - \; \frac{[\tau + x_2 - x_1][\gamma]}{[x_2 - x_1][\tau - \gamma]} \mathcal{A}(x_1, \tau) \mathcal{B}(x_2, \tau - \gamma) \; .
\>
In particular, the iteration of $\mathscr{S}_{\mathcal{A}, \mathcal{B}}$ yields the following relation in $\mathscr{A}_{L+1} (\mathcal{R})$,
\< \label{AL1}
&& \mathcal{A} (x_0, \tau + 2\gamma) \; Y_{\tau} (X) = \nonumber \\
&& \frac{[\tau + 2\gamma]}{[\tau + (L+2)\gamma]} \prod_{j=1}^L \frac{[x_j - x_0 + \gamma]}{[x_j - x_0]} \; Y_{\tau + \gamma} (X) \; \mathcal{A} (x_0, \tau + (L+2)\gamma) \nonumber \\
&& + \; \sum_{i=1}^L \frac{[\tau + 2\gamma + x_0 - x_i][\gamma]}{[x_0 - x_i][\tau + (L+2)\gamma]} \prod_{\substack{j=1 \\ j \neq i}}^{L} \frac{[x_j - x_i + \gamma]}{[x_j - x_i]} \; Y_{\tau + \gamma} (X_i^0) \; \mathcal{A} (x_i, \tau + (L+2)\gamma) \; , \nonumber \\
\>
where $X \coloneqq \{ x_1 , x_2 , \dots , x_L \}$, $X_i^{\alpha} \coloneqq X \cup \{ x_{\alpha} \} \backslash \{ x_i \}$ and 
$Y_{\tau} (X) \coloneqq \PROD{1}{j}{L} \mathcal{B} (x_j, \tau + j \gamma)$.

\begin{thm}[Equation type A] \label{thmA} 
The partition function $Z_{\tau}$ satisfies the functional equation 
\[
\label{eqA}
M_0^{(\mathcal{A})} \; Z_{\tau} (X) + \sum_{i \in \{0, 1, 2, \dots , L \}} N_i^{(\mathcal{A})} \ Z_{\tau + \gamma} (X_i^0) = 0 \; ,
\]
with coefficients given by
\< \label{coeffA}
M_0^{(\mathcal{A})} &\coloneqq& \frac{[\tau + \gamma]}{[\tau + (L+1)\gamma]} \prod_{j=1}^{L} [x_0 - \mu_j] \nonumber \\
N_0^{(\mathcal{A})} &\coloneqq& -\frac{[\tau + 2\gamma]}{[\tau + (L+2)\gamma]} \prod_{j=1}^{L} [x_0 - \mu_j + \gamma] \prod_{j=1}^{L} \frac{[x_j - x_0 + \gamma]}{[x_j - x_0]}  \nonumber \\
N_i^{(\mathcal{A})} &\coloneqq& \frac{[\tau + 2\gamma + x_0 - x_i] [\tau]}{[x_i - x_0][\tau + (L+2)\gamma]} \prod_{j=1}^{L} [x_i - \mu_j + \gamma] \prod_{\underset{j \neq i}{j=1}}^L \frac{[x_j - x_i + \gamma]}{[x_j - x_i]} \qquad i = 1,2, \dots , L \; . \nonumber \\
\>
\end{thm}
\begin{proof}
We firstly remark that $Y_{\tau} (X) \in \mathscr{A}_{L} (\mathcal{R})$ and apply the functional $\Phi$ given by \eqref{O} onto \eqref{AL1}. Then we only
need to identify $\Phi \left( Y_{\tau} (X)  \right) = Z_{\tau} (X)$ using \eqref{pf}. 
\end{proof}

\begin{rema}
Eq. \eqref{eqA} is the very same functional equation derived in \cite{Galleas_2013}. The precise matching can be achieved 
by performing a trivial shift in the dynamical parameter.
\end{rema}

\subsection{Equation type D} \label{sec:typeD}

Formula \eqref{AL1} is not the only potentially useful relation in $\mathscr{A}_{L+1} (\mathcal{R})$. In fact, one can find a large number
of relations in $\mathscr{A}_{n} (\mathcal{R})$ $ ( n \geq L+1 )$ which can be exploited along the lines of the AF framework in order to describe
the partition function \eqref{pf}. However, the simplest ones considered so far seem to have structure resembling \eqref{eqA}. 
One of the goals of the present paper is to obtain determinantal representations for $Z_{\tau}$ and for that we will also make use of another equation
somehow similar to \eqref{eqA}. We shall refer to such equation as \emph{equation type D} and its derivation within the AF framework exploits the subalgebra
$\mathscr{S}_{\mathcal{D}, \mathcal{B}, H} \subset \mathscr{A}_{2} (\mathcal{R}) $ formed by
\begin{align}
\mathcal{D} (x_1, \tau) f(H) &= f(H) \; \mathcal{D} (x_1, \tau) & \mathcal{D}(x_1, \tau) \; \mathcal{D}(x_2, \tau + \gamma) &= \mathcal{D}(x_2, \tau) \; \mathcal{D}(x_1, \tau + \gamma)   \nonumber \\
\mathcal{B} (x_1, \tau) f(H) &= f(H+2) \; \mathcal{B} (x_1, \tau) & \mathcal{B}(x_1, \tau) \; \mathcal{B}(x_2, \tau + \gamma) &= \mathcal{B}(x_2, \tau) \; \mathcal{B}(x_1, \tau + \gamma) \nonumber 
\end{align}
\< 
&& [\tau + \gamma (1 - H)] \; \mathcal{D}(x_1, \tau) \mathcal{B}(x_2, \tau + \gamma) = \nonumber \\
&& \qquad \qquad \qquad \quad \frac{[x_1 - x_2 + \gamma]}{[x_1 - x_2]} [\tau - \gamma H] \; \mathcal{B}(x_2, \tau ) \mathcal{D}(x_1, \tau + \gamma) \nonumber \\
&& \qquad \qquad \qquad \quad - \frac{[\gamma]}{[x_1 - x_2]} [\tau + x_1 - x_2 - \gamma H] \; \mathcal{B} (x_1 , \tau) \mathcal{D} (x_2, \tau + \gamma)  \nonumber 
\> 
\< \label{SDBH}
&& [\tau - \gamma (1 + H)] \; \mathcal{B}(x_1, \tau) \mathcal{D}(x_2, \tau + \gamma) = \nonumber \\
&& \qquad \qquad \qquad \quad \frac{[x_1 - x_2 + \gamma]}{[x_1 - x_2]} [\tau - \gamma H] \; \mathcal{D}(x_2, \tau ) \mathcal{B}(x_1, \tau + \gamma) \nonumber \\ 
&& \qquad \qquad \qquad \quad - \frac{[\gamma]}{[x_1 - x_2]} [\tau + x_2 - x_1 - \gamma H] \; \mathcal{D}(x_1, \tau) \mathcal{B}(x_2, \tau + \gamma) \; . \nonumber \\
\>
The subalgebra $\mathscr{S}_{\mathcal{D}, \mathcal{B}, H}$ is generated by $\mathcal{D}$, $\mathcal{B}$ and the Cartan element $H$; and in order to
proceed we also need an analogous of \eqref{AL1} associated to the iteration of $\mathscr{S}_{\mathcal{D}, \mathcal{B}, H}$.
Similarly to derivation presented in \Secref{sec:typeA}, the iteration of $\mathscr{S}_{\mathcal{D}, \mathcal{B}, H}$ produces the following relation
in $\mathscr{A}_{L+1} (\mathcal{R})$,
\< \label{DL1}
&& \mathcal{D} (x_{\bar{0}}, \tau + \gamma ) \; Y_{\tau + \gamma} (X) = \nonumber \\
&& [\tau + 2\gamma] \prod_{j=1}^L \frac{[x_{\bar{0}} - x_j + \gamma]}{[x_{\bar{0}} - x_j]} \; Y_{\tau} (X) \; \mathcal{D} (x_{\bar{0}}, \tau + \gamma (L+1) ) \; [\tau + \gamma (H+2)]^{-1} \nonumber \\
&& - \sum_{i=1}^L \frac{[\gamma] [\tau + 2\gamma]}{[x_{\bar{0}} - x_i]} \prod_{\substack{j=1 \\ j \neq i}}^L \frac{[x_i - x_j + \gamma]}{[x_i - x_j]} \; Y_{\tau} (X_i^{\bar{0}}) \; \mathcal{D} (x_i , \tau + \gamma (L+1)) \; [\tau + \gamma (H+1) + x_{\bar{0}} - x_i] \nonumber \\
&& \qquad \qquad \qquad \qquad \qquad \qquad \qquad \qquad \qquad \quad  \times [\tau + \gamma (H+1)]^{-1} [\tau + \gamma (H+2)]^{-1} \; . \nonumber \\
\>

\begin{thm}[Equation type D] \label{thmD}
The functional relation
\[
\label{eqD}
M_0^{(\mathcal{D})} \; Z_{\tau + \gamma} (X) + \sum_{i \in \{ \bar{0}, 1, 2, \dots, L \} } N_i^{(\mathcal{D})} \; Z_{\tau} (X_i^{\bar{0}}) = 0 \; ,
\]
with coefficients reading
\< \label{coeffD}
M_0^{(\mathcal{D})} &\coloneqq& \prod_{j=1}^{L} [x_{\bar{0}} - \mu_j + \gamma] \nonumber \\
N_0^{(\mathcal{D})} &\coloneqq& - \prod_{j=1}^{L} [x_{\bar{0}} - \mu_j] \prod_{j=1}^{L} \frac{[x_{\bar{0}} - x_j + \gamma]}{[x_{\bar{0}} - x_j]}  \nonumber \\
N_i^{(\mathcal{D})} &\coloneqq& \frac{[\gamma] [\tau + (L+1)\gamma + x_{\bar{0}} - x_i]}{[x_{\bar{0}} - x_i] [\tau + (L+1)\gamma]} \prod_{j=1}^{L} [x_i - \mu_j] \prod_{\substack{j=1 \\ j \neq i}}^L \frac{[x_i - x_j + \gamma]}{[x_i - x_j]} \qquad i = 1,2, \dots , L \; , \nonumber \\
\>
is satisfied by the partition function \eqref{pf}.
\end{thm}
\begin{proof}
Same as for Theorem \ref{thmA} but using \eqref{DL1} instead of \eqref{AL1}.
\end{proof}

\begin{rema}
The terminology \emph{Equations type A and D}, respectively associated to \eqref{eqA} and \eqref{eqD}, are intrinsically associated to 
their roots in relations \eqref{AL1} and \eqref{DL1}.
\end{rema}

\section{Modified functional relations} \label{sec:AD}

In \cite{Galleas_2013} we have shown how equation type A can be solved in terms of a multiple contour integral. A similar analysis can 
be performed on equation type D yielding a comparable representation. Here, however, we shall employ a different methodology for analyzing 
equations type A and D. In particular, our goal is to obtain the solution $Z_{\tau}$ in terms of a determinant and we anticipate such representation emerges naturally from a modified version of equations \eqref{eqA} and \eqref{eqD}.   

\begin{thm} \label{thmAD}
The partition function \eqref{pf} fulfills the relation
\[
\label{eqADneu}
\mathcal{M}_0 \; Z_{\tau} (X) + \sum_{i=1}^L \mathcal{N}_i \; Z_{\tau} (X_i^0) + \sum_{i=1}^L \bar{\mathcal{N}}_i \; Z_{\tau} (X_i^{\bar{0}}) = 0 \; ,
\]
with coefficients defined as
\<
\label{coeffADneu} 
\mathcal{M}_0 &\coloneqq& \prod_{j=1}^L \frac{[x_0 - x_j + \gamma] [x_0 - \mu_j] [x_{\bar{0}} - \mu_j + \gamma]}{[x_0 - x_j] [x_0 - \mu_j + \gamma]} - \prod_{j=1}^L \frac{[x_{\bar{0}} - x_j + \gamma] [x_{\bar{0}} - \mu_j]}{[x_{\bar{0}} - x_j ]} \nonumber \\
\mathcal{N}_i &\coloneqq& - \frac{[\gamma] [x_0 - x_i + \tau + (L+1)\gamma]}{[\tau + (L+1)\gamma] [x_0 - x_i]} \prod_{j=1}^L \frac{[x_i - \mu_j] [x_{\bar{0}} - \mu_j + \gamma]}{[x_0 - \mu_j + \gamma]} \prod_{\substack{j=1 \\ j \neq i}}^L \frac{[x_i - x_j + \gamma]}{[x_i - x_j]} \nonumber \\
\bar{\mathcal{N}}_i &\coloneqq& \frac{[\gamma] [x_{\bar{0}} - x_i + \tau + (L+1)\gamma]}{[\tau + (L+1)\gamma] [x_{\bar{0}} - x_i]} \prod_{j=1}^L [x_i - \mu_j] \prod_{\substack{j=1 \\ j \neq i}}^L \frac{[x_i - x_j + \gamma]}{[x_i - x_j]} \; . \nonumber \\
\>
\end{thm}
\begin{proof}

Equation \eqref{eqADneu} follows from a reformulation of equation \eqref{eqD}. We start by remarking equation \eqref{eqD} depends on $L+1$ spectral variables, namely the set $\{x_{\bar{0}}, x_1, x_2, \dots, x_L \}$, whilst the function $Z_{\tau}$ depends on only $L$ variables. Next we conveniently rewrite \eqref{eqD} as
\[ \label{temp}
Z_{\tau + \gamma} (X) =  - \sum_{i \in \{ \bar{0}, 1, 2, \dots, L \} } \frac{N_i^{(\mathcal{D})} (x_{\bar{0}}, x_1, x_2, \dots, x_L) }{M_0^{(\mathcal{D})} (x_{\bar{0}}, x_1, x_2, \dots, x_L)} \; Z_{\tau} (X_i^{\bar{0}}) 
\]
including the dependence of the coefficients $M_0^{(\mathcal{D})}$ and $N_i^{(\mathcal{D})}$ on the variables $x_{\bar{0}}, x_1, x_2, \dots, x_L$.
From expression \eqref{temp} we can now conclude its RHS does not depend on the variable $x_{\bar{0}}$. Therefore, one can write
\< \label{temp1}
\sum_{i \in \{ \bar{0}, 1, 2, \dots, L \} } \frac{N_i^{(\mathcal{D})} (x_{\bar{0}}, x_1, x_2, \dots, x_L) }{M_0^{(\mathcal{D})} (x_{\bar{0}}, x_1, x_2, \dots, x_L)} \; Z_{\tau} (X_i^{\bar{0}}) = \sum_{i \in \{ 0, 1, 2, \dots, L \} } \frac{N_i^{(\mathcal{D})} (x_{0}, x_1, x_2, \dots, x_L) }{M_0^{(\mathcal{D})} (x_{0}, x_1, x_2, \dots, x_L)} \; Z_{\tau} (X_i^{0}) \; . \nonumber \\
\>
Eq. \eqref{eqADneu} is then identified with \eqref{temp1} upon the multiplication by an overall factor. This completes our proof.
\end{proof}

\begin{rema}
Equation \eqref{eqADneu} is the same equation derived previously in \cite{Galleas_pre_2016}, although in \cite{Galleas_pre_2016} we have obtained it from a particular combination of the equations type A and type D.
\end{rema}

\begin{rema}
The mechanism leading to \eqref{eqADneu} has only made use of equation \eqref{eqD}. Another equation with the same structure of \eqref{eqADneu} but with different coefficients  could also be obtained from the same analysis using equation \eqref{eqA} instead of \eqref{eqD}.
\end{rema}

Theorem \ref{thmAD} has important implications that we are now ready to discuss. For instance, when comparing \eqref{eqA} and \eqref{eqD} with equation
\eqref{eqADneu} one can readily notice that the dynamical parameter $\tau$ no longer plays the role of variable. All terms in 
\eqref{eqADneu} are evaluated at the same value of $\tau$ and it can be fixed from this point on. On the other hand, the structure of
\eqref{eqADneu} is more complicated than that of equations \eqref{eqA} and \eqref{eqD} as it now depends on two extra variables, namely $x_0$ and $x_{\bar{0}}$,
in addition to the set $\{ x_1, x_2 , \dots , x_L \}$ required to describe the partition function $Z_{\tau}$.
However, we shall see that these additional variables play an important role for solving \eqref{eqADneu} in terms of a determinant.

\section{Partition function evaluation} \label{sec:PF}

In \Secref{sec:AD} we have obtained two functional equations, namely \eqref{eqA} and \eqref{eqD}, characterizing the partition function \eqref{pf}. In principle one can solve directly the aforementioned equations as shown in \cite{Galleas_2013} but here we shall consider a different approach. In this work we are interested in expressing
$Z_{\tau}$ as a determinant and this seems to be doable through the method recently put forward in \cite{Galleas_2016}. However, the dependence of \eqref{eqA} and \eqref{eqD} on the dynamical parameter $\tau$ seems to prevent the direct use of the approach described in \cite{Galleas_2016}. Eq. \eqref{eqADneu}, on the other hand, has no issues with the dynamical parameter it and can be tackled through a generalization of the method of \cite{Galleas_2016}. That is the main reason for considering
\eqref{eqADneu} instead of \eqref{eqA} or \eqref{eqD}.
More precisely, Eq. \eqref{eqADneu} downgrades the status of the dynamical parameter at the expense of introducing an extra spectral variable in addition to the $L+1$ already present in \eqref{eqA} and \eqref{eqD}. Although this additional variable makes the structure of \eqref{eqADneu} more involved when compared to \eqref{eqA} and \eqref{eqD}, it will also help us in expressing $Z_{\tau}$ as a determinant.

We then start our analysis by remarking some features of Eq. \eqref{eqADneu} originated from the distinct role played by the variables 
$x_0$ and $x_{\bar{0}}$. For instance, at first sight it seems that \eqref{eqADneu} reduces to the same type
of functional equation obtained for the six-vertex model in \cite{Galleas_proc, Galleas_2013} under the specialization $x_0 = x_{\bar{0}}$.
However, this is not the case and one can readily see that Eq. \eqref{eqADneu} is automatically satisfied when this particular condition
is fulfilled. Hence, we eliminate this point from our analysis.
The next step is to study the behavior of the equation \eqref{eqADneu} under permutations of variables; and for that it is convenient to introduce the following
conventions.
\begin{defi}[Permutation] \label{pij}
Let $n \in \Z_{>0}$ and $\mathcal{S}_n$ be the symmetric group in $n$ letters acting by permutations on $\C^n$. Also, for $1 \leq i, j \leq n $
write $\Pi_{i,j} \in \mathcal{S}_n$ for the permutation of letters $i$ and $j$; and let $\text{Fun}(\C^n)$ be the space of meromorphic
functions on $\C^n$. The symmetric group $\mathcal{S}_n$ then acts on $f \in \text{Fun}(\C^n)$ by
\[
\left( \Pi_{i,j} f \right) (x_1, \dots , x_i , \dots , x_j , \dots , x_n) \coloneqq f(x_1, \dots , x_j , \dots , x_i , \dots , x_n) \; . \nonumber 
\]
\end{defi}

Then, considering Definition \ref{pij}, the first important observation is that \eqref{eqADneu} is invariant under the action of $\Pi_{i,j}$ for $1 \leq i, j \leq L$. On the other hand, permutations $\Pi_{0, l}$ and $\Pi_{\bar{0}, m}$ for $0 \leq l, m \leq L$ produces new equations with structure not being captured by \eqref{eqADneu}. Nevertheless, these new equations are of the form
\<
\label{eqADper}
\mathcal{M}_0^{(l,m)} \; Z_{\tau} (X) + \sum_{i=1}^L \mathcal{N}_i^{(l,m)} \; Z_{\tau} (X_i^0) + \sum_{i=1}^L \bar{\mathcal{N}}_i^{(l,m)} \; Z_{\tau} (X_i^{\bar{0}}) 
+ \sum_{1 \leq i, j \leq L} \mathcal{O}_{ij}^{(l,m)} \; Z_{\tau} ( X_{i,j}^{0, \bar{0}} ) = 0 \; , \nonumber \\
\>
with symbol $X_{i,j}^{\alpha, \beta} \coloneqq X \cup \{x_{\alpha}, x_{\beta} \} \backslash \{x_i, x_j \}$ additionaly introduced.
The label $(l,m)$ in the coefficients of \eqref{eqADper} also requires some clarifications. The coefficients with label $(l,m) = (l, \bar{0})$ 
arise from \eqref{eqADneu} under permutation $\Pi_{0, l}$ while keeping $x_{\bar{0}}$ fixed. Similarly, Eq. \eqref{eqADper}
with label $(l,m) = (0, m)$ results from \eqref{eqADneu} under the action of $\Pi_{\bar{0}, m}$ while $x_0$ is now kept fixed.
As for the remaining cases; that is \eqref{eqADper} with label $(l,m)$ restricted to $1 \leq l < m \leq L$, they are obtained through
permutations $\Pi_{0, l}$ and $\Pi_{\bar{0}, m}$ in this respective order. Taking into account the aforementioned permutation of 
variables, we are then left with the following expressions for the coefficients in \eqref{eqADper}:

\begin{align} \label{l0}
\mathcal{M}_0^{(l, \bar{0})} & \coloneqq \Pi_{0, l} \; \mathcal{N}_l  & \bar{\mathcal{N}}_j^{(l, \bar{0})} & \coloneqq \begin{cases}  \Pi_{0, l} \; \bar{\mathcal{N}}_l \qquad j=l \nonumber \\ 0 \qquad \qquad  \mbox{otherwise} \end{cases} \nonumber \\
\mathcal{N}_j^{(l, \bar{0})} & \coloneqq \begin{cases}  \Pi_{0, l} \; \mathcal{M}_0  \qquad j=l \nonumber \\ \Pi_{0, l} \; \mathcal{N}_j  \qquad \mbox{otherwise} \end{cases} & \mathcal{O}_{ij}^{(l, \bar{0})} & \coloneqq \begin{cases}  \Pi_{0, l} \; \bar{\mathcal{N}}_i \qquad j=l \nonumber \\ \Pi_{0, l} \; \bar{\mathcal{N}}_j  \qquad i=l \nonumber \\ 0 \qquad \qquad \mbox{otherwise} \end{cases} \nonumber \\
\end{align}

\begin{align} \label{0m}
\mathcal{M}_0^{(0, m)} & \coloneqq \Pi_{\bar{0}, m} \; \bar{\mathcal{N}}_m  & \bar{\mathcal{N}}_j^{(0, m)} & \coloneqq \begin{cases}  \Pi_{\bar{0}, m} \; \mathcal{M}_0 \qquad j=m \nonumber \\ \Pi_{\bar{0}, m} \; \bar{\mathcal{N}}_j  \qquad  \mbox{otherwise} \end{cases} \nonumber \\
\mathcal{N}_j^{(0, m)} & \coloneqq \begin{cases}  \Pi_{\bar{0}, m} \; \mathcal{N}_m  \qquad j=m \nonumber \\ 0 \qquad \qquad \mbox{otherwise} \end{cases} & \mathcal{O}_{ij}^{(0, m)} & \coloneqq \begin{cases}  \Pi_{\bar{0}, m} \; \mathcal{N}_i  \qquad j=m \nonumber \\ \Pi_{\bar{0}, m} \; \mathcal{N}_j \qquad i=m \nonumber \\ 0 \qquad \qquad \;\;\;  \mbox{otherwise} \end{cases} \nonumber \\
\end{align}

\begin{align} \label{lm}
\mathcal{M}_0^{(l, m)} & \coloneqq 0 & \bar{\mathcal{N}}_j^{(l, m)} & \coloneqq \begin{cases}  \Pi_{\bar{0}, m} \circ \Pi_{0, l} \; \bar{\mathcal{N}}_l  \qquad j=l \nonumber \\ \Pi_{\bar{0}, m} \circ \Pi_{0, l} \; \mathcal{N}_l  \qquad  j=m \nonumber \\ 0 \qquad \qquad \qquad \quad \mbox{otherwise} \end{cases} \nonumber \\
\mathcal{N}_j^{(l, m)} & \coloneqq \begin{cases}  \Pi_{\bar{0}, m} \circ \Pi_{0, l} \; \bar{\mathcal{N}}_m  \qquad j=l \nonumber \\ \Pi_{\bar{0}, m} \circ \Pi_{0, l} \; \mathcal{N}_m  \qquad j=m \nonumber \\ 0 \qquad \qquad \qquad \;\; \mbox{otherwise} \end{cases} & \mathcal{O}_{ij}^{(l, m)} & \coloneqq \begin{cases} \Pi_{\bar{0}, m} \circ \Pi_{0, l} \; \mathcal{M}_0  \qquad i=l, j=m \nonumber \\ \Pi_{\bar{0}, m} \circ \Pi_{0, l} \; \bar{\mathcal{N}}_j  \qquad \; i=l, j \neq m \nonumber \\ \Pi_{\bar{0}, m} \circ \Pi_{0, l} \; \mathcal{N}_j \qquad \; i=m \nonumber \\ \Pi_{\bar{0}, m} \circ \Pi_{0, l} \; \bar{\mathcal{N}}_i \qquad \; j=l \nonumber \\ \Pi_{\bar{0}, m} \circ \Pi_{0, l} \; \mathcal{N}_i  \qquad \; j=m, i \neq l \nonumber \\ 0 \qquad \qquad \quad \quad \quad \;\; \mbox{otherwise} \end{cases} \; . \nonumber \\
\end{align}

\begin{rema}
The structure of \eqref{eqADper} resembles that of the equation presented in \cite{Galleas_2012} for the trigonometric \textsc{sos} model
with domain-wall boundaries. However, there are still crucial differences and the existence of a precise relation is not clear at the 
moment.
\end{rema}

In order to proceed let us have a closer look into the system of equations \eqref{eqADper}. One can clearly see this system comprises
$d_L \coloneqq L (L+3)/2$ equations: $L$ equations coming from permutations $\Pi_{0, l}$, another $L$ equations
originated from the action of $\Pi_{\bar{0}, m}$ and $L (L-1)/2$ obtained from the composed action  $\Pi_{\bar{0}, m} \circ \Pi_{0, l}$
under the restriction $1 \leq l < m \leq L$. In complement to that we can see the system of Eqs. \eqref{eqADper} encloses $d_L +1$ unknowns. They are the $L$ terms of the form $Z_{\tau} (X_i^0)$, another $L$ of 
type $Z_{\tau} (X_i^{\bar{0}})$, $L (L-1)/2$ terms recasted as $Z_{\tau} ( X_{i,j}^{0, \bar{0}} )$ and one term $Z_{\tau} (X)$.
In this way, assuming equations \eqref{eqADper} are linearly independent, we have enough equations for expressing each function 
$Z_{\tau} (X_i^0)$, $Z_{\tau} (X_i^{\bar{0}})$ and $Z_{\tau} ( X_{i,j}^{0, \bar{0}} )$ in terms of $Z_{\tau} (X)$. 
The issue of linear independence of equations \eqref{eqADper} will be addressed later and for the moment we will simply assume it holds. Therefore, using Cramer's method we find
\<
\label{ZZ}
Z_{\tau} (X_i^0) = \frac{\left| \begin{matrix} \mathcal{F}_i & \mathcal{I} & \mathcal{G} \cr \bar{\mathcal{I}}_i & \mathcal{K} & \mathcal{J} \cr \bar{\mathcal{F}}_i & \bar{\mathcal{J}} & \bar{\mathcal{G}} \end{matrix} \right|}{\left| \begin{matrix} \mathcal{F} & \mathcal{I} & \mathcal{G} \cr \bar{\mathcal{I}} & \mathcal{K} & \mathcal{J} \cr \bar{\mathcal{F}} & \bar{\mathcal{J}} & \bar{\mathcal{G}} \end{matrix} \right|} Z_{\tau} (X)  \quad , \quad  
Z_{\tau} (X_i^{\bar{0}}) = \frac{\left| \begin{matrix} \mathcal{F} & \mathcal{I} & \mathcal{G}_i \cr \bar{\mathcal{I}} & \mathcal{K} & \mathcal{J}_i \cr \bar{\mathcal{F}} & \bar{\mathcal{J}} & \bar{\mathcal{G}}_i \end{matrix} \right|}{\left| \begin{matrix} \mathcal{F} & \mathcal{I} & \mathcal{G} \cr \bar{\mathcal{I}} & \mathcal{K} & \mathcal{J} \cr \bar{\mathcal{F}} & \bar{\mathcal{J}} & \bar{\mathcal{G}} \end{matrix} \right|} Z_{\tau} (X) 
\>
and 
\<
\label{ZZZ}
Z_{\tau} ( X_{i,j}^{0, \bar{0}} ) =  \frac{\left| \begin{matrix} \mathcal{F} & \mathcal{I}_{ij} & \mathcal{G} \cr \bar{\mathcal{I}} & \mathcal{K}_{ij} & \mathcal{J} \cr \bar{\mathcal{F}} & \bar{\mathcal{J}}_{ij} & \bar{\mathcal{G}} \end{matrix} \right|}{\left| \begin{matrix} \mathcal{F} & \mathcal{I} & \mathcal{G} \cr \bar{\mathcal{I}} & \mathcal{K} & \mathcal{J} \cr \bar{\mathcal{F}} & \bar{\mathcal{J}} & \bar{\mathcal{G}} \end{matrix} \right|} Z_{\tau} (X)  \; .
\>
The matrix coefficients entering the determinants in \eqref{ZZ} and \eqref{ZZZ} follow directly from Eq. \eqref{eqADper}. More precisely, the coefficients $\mathcal{F}$, $\bar{\mathcal{F}}$, 
$\mathcal{G}$ and $\bar{\mathcal{G}}$ are submatrices of dimension $L \times L$ with entries defined as
\begin{align} \label{FG}
\mathcal{F}_{a,b} &\coloneqq \mathcal{N}_b^{(a, \bar{0})} & \mathcal{G}_{a,b} &\coloneqq \bar{\mathcal{N}}_b^{(a, \bar{0})} \nonumber \\
\bar{\mathcal{F}}_{a,b} &\coloneqq \mathcal{N}_b^{(0, a)} & \bar{\mathcal{G}}_{a,b} &\coloneqq \bar{\mathcal{N}}_b^{(0, a)} \; .
\end{align}

As for the remaining submatrices, it is convenient to introduce an index $n \colon \Z \times \Z \to \Z$ defined explicitly
as $n_{r,s} \coloneqq s+L(r-1)-\frac{r(r+1)}{2}$ in the domain $1 \leq r < s \leq L$. In this way, the matrices $\mathcal{I}$ and $\bar{\mathcal{J}}$ are of dimension
$L \times \frac{L(L-1)}{2}$ with entries defined as
\< \label{IJb}
\mathcal{I}_{a,n_{r,s}} \coloneqq \mathcal{O}_{r s}^{(a, \bar{0})} \qquad \text{and} \qquad \bar{\mathcal{J}}_{a, n_{r,s}} \coloneqq \mathcal{O}_{r s}^{(0, a)} \; .
\>
Next we turn our attention to the matrices $\bar{\mathcal{I}}$ and $\mathcal{J}$. Their entries are defined as
\< \label{IbJ}
\bar{\mathcal{I}}_{n_{l,m}, b} \coloneqq \mathcal{N}_b^{(l,m)} \qquad \text{and} \qquad \mathcal{J}_{n_{l,m}, b} \coloneqq \bar{\mathcal{N}}_b^{(l,m)} \; .
\>
The latter definition builds up matrices of dimension $\frac{L(L-1)}{2} \times L$. In its turn the matrix $\mathcal{K}$
has dimension $\frac{L(L-1)}{2} \times \frac{L(L-1)}{2}$ and its entries are defined as
\< \label{KK}
\mathcal{K}_{n_{l, m}, n_{r, s}} \coloneqq \mathcal{O}_{rs}^{(l,m)} \; .
\>
Here we recall that $1 \leq l < m \leq L$ and $1 \leq r < s \leq L$.

Definitions \eqref{FG}-\eqref{KK} allow one to write down only the matrix appearing in the denominator of relations \eqref{ZZ} and 
\eqref{ZZZ}. As for the remaining matrices appearing in the numerators we still need to define the tuples of matrices
$(\mathcal{F}_i, \bar{\mathcal{I}}_i , \bar{\mathcal{F}}_i)$, $(\mathcal{G}_i, \mathcal{J}_i , \bar{\mathcal{G}}_i)$ 
and $(\mathcal{I}_{ij}, \mathcal{K}_{ij} , \bar{\mathcal{J}}_{ij})$. Expressions \eqref{ZZ} and \eqref{ZZZ} are obtained from Cramer's rule
and consequently the remaining matrices do not differ drastically from the ones already presented in \eqref{FG}-\eqref{KK}. In this way, we have the following entries:
\< \label{FI}
(\mathcal{F}_i)_{a,b} \coloneqq \begin{cases} - \mathcal{M}_0^{(a, \bar{0})} \qquad \qquad i=b \\ \mathcal{N}_b^{(a, \bar{0})} \qquad \qquad \mbox{otherwise} \end{cases} \qquad , \qquad (\bar{\mathcal{I}}_i)_{n_{l,m}, b} &\coloneqq& \begin{cases} - \mathcal{M}_0^{(l, m)} \qquad \qquad i=b \\ \mathcal{N}_b^{(l,m)} \qquad \qquad \mbox{otherwise} \end{cases} \nonumber \\
\>
and
\< \label{FBI}
(\bar{\mathcal{F}}_i)_{a,b} \coloneqq \begin{cases} - \mathcal{M}_0^{(0,a)} \qquad \qquad i=b \\ \mathcal{N}_b^{(0,a)} \qquad \qquad \mbox{otherwise} \end{cases} \; ; \nonumber \\
\>
for the matrices in the first tuple $(\mathcal{F}_i, \bar{\mathcal{I}}_i , \bar{\mathcal{F}}_i)$.
As for the tuple $(\mathcal{G}_i, \mathcal{J}_i, \bar{\mathcal{G}}_i )$ we have
\< \label{GJi}
(\mathcal{G}_i)_{a,b} \coloneqq \begin{cases} - \mathcal{M}_0^{(a, \bar{0})} \qquad \qquad i=b \\ \bar{\mathcal{N}}_b^{(a, \bar{0})} \qquad \qquad \mbox{otherwise} \end{cases} \qquad , \qquad (\mathcal{J}_i)_{n_{l,m}, b} &\coloneqq& \begin{cases} - \mathcal{M}_0^{(l, m)} \qquad \qquad i=b \\ \bar{\mathcal{N}}_b^{(l,m)} \qquad \qquad \mbox{otherwise} \end{cases}  \nonumber \\
\>
and
\< \label{Gbi}
(\bar{\mathcal{G}}_i)_{a,b} \coloneqq \begin{cases} - \mathcal{M}_0^{(0,a)} \qquad \qquad i=b \\ \bar{\mathcal{N}}_b^{(0,a)} \qquad \qquad \mbox{otherwise} \end{cases} \; . \nonumber \\
\>
Lastly, the elements of the tuple $(\mathcal{I}_{ij}, \mathcal{K}_{ij} , \bar{\mathcal{J}}_{ij})$ are built with the help of the following
definitions:
\< \label{IKJbij}
(\mathcal{I}_{ij})_{a,n_{r,s}} &\coloneqq& \begin{cases} - \mathcal{M}_0^{(a, \bar{0})} \qquad \qquad i=r, j=s \\ \mathcal{O}_{r s}^{(a, \bar{0})} \qquad \qquad \qquad \mbox{otherwise} \end{cases} \nonumber \\
(\mathcal{K}_{ij})_{n_{l,m}, n_{r,s}} &\coloneqq& \begin{cases} - \mathcal{M}_0^{(l,m)} \qquad \qquad i=r, j=s \\ \mathcal{O}_{r s}^{(l,m)} \qquad \qquad \qquad \mbox{otherwise} \end{cases}  \nonumber \\
(\bar{\mathcal{J}}_{ij})_{a,n_{r,s}} &\coloneqq& \begin{cases} - \mathcal{M}_0^{(0,a)} \qquad \qquad i=r, j=s \\ \mathcal{O}_{r s}^{(0,a)} \qquad \qquad \qquad \mbox{otherwise} \end{cases} \; . \nonumber \\
\>
Here we stress again our relations \eqref{FI}-\eqref{IKJbij} are defined in the domains $1 \leq l < m \leq L$ and $1 \leq r < s \leq L$. Also, from  \eqref{l0}-\eqref{lm} one can notice some coefficients in Eq. \eqref{eqADper} vanish and, consequently, the determinants appearing in \eqref{ZZ}
and \eqref{ZZZ} are taken over sparse matrices.

As previously remarked, the possibility of writing \eqref{ZZ} and \eqref{ZZZ} depends on the assumption that equations \eqref{eqADper} are linearly independent.  This latter hypothesis will be addressed in more details in the following Lemma.   

\begin{lem}
Equations \eqref{eqADper} are linearly independent.
\end{lem}

\begin{proof}
In order to prove our statement one can show the matrix of coefficients associated to \eqref{eqADper} has non-vanishing determinant for generic values of its parameters. The corresponding matrix of coefficients is given by
\< 
\Omega^{(L)}_{\tau} \coloneqq \begin{pmatrix} \mathcal{F} & \mathcal{I} & \mathcal{G} \cr \bar{\mathcal{I}} & \mathcal{K} & \mathcal{J} \cr \bar{\mathcal{F}} & \bar{\mathcal{J}} & \bar{\mathcal{G}} \end{pmatrix} 
\>
with superscript and subscript emphasizing respectively its dependence on the chain length $L$ and dynamical parameter $\tau$. 
More precisely, the matrix $\Omega^{(L)}_{\tau}$ depends on the spectral parameters $x_{\bar{0}}$, $x_0$, $x_j \in \C$ $(1 \leq j \leq L)$ in addition to the inhomogeneity parameters $\mu_j \in \C$ $(1 \leq j \leq L)$ and $\tau$, $\gamma \in \C$. In fact, it also depends on the elliptic nome $0 < p < 1$ but this dependence will not be relevant in our analysis.

We want to show $\mathrm{det} ( \Omega^{(L)}_{\tau} )$ does not vanish identically; and that can be proved by contradiction as follows. Suppose 
$\mathrm{det} ( \Omega^{(L)}_{\tau} ) = 0$ for generic values of the parameters $x_{\bar{0}}$, $x_0$ and $x_j$ with $1 \leq j \leq L$. Therefore, we also have 
$\mathrm{det} ( \Omega^{(L)}_{\tau} ) = 0$  for specializations of the aforementioned parameters.
However, the specialization
\< \label{spec}
x_{\bar{0}} &=& x_0 + \tau + (L+1)\gamma  \nonumber \\
x_L &=& \mu_L \nonumber \\
x_j &=& x_0 + j \gamma \qquad 1 \leq j \leq L-1 
\>
causes some rows and columns of $\Omega^{(L)}_{\tau}$ to have a single non-null entry; allowing a simple expansion of $\mathrm{det} ( \Omega^{(L)}_{\tau} )$ by minors. In this way, using repeatedly determinant expansion by minors, we arrive at the expression
\< \label{R1}
\mathrm{det} ( \left. \Omega^{(L)}_{\tau} \right|_{*} ) = (-1)^{L+1} U_{L,L} U_{2L-1,d_L} U_{d_L, \frac{L(L+1)}{2}} \prod_{n=3}^L U_{(2L+1-n) \frac{n}{2}, (2L+2-n) \frac{(n-1)}{2}} \;  \mathrm{det} (\mathrm{T}_{\Omega}) \nonumber \\
\>
with symbol $|_{*}$ denoting the specialization \eqref{spec}, $U_{i,j}$ the entries of $\Omega^{(L)}_{\tau} |_{*}$ and matrix $\mathrm{T}_{\Omega}$ 
resembling $\Omega^{(L-1)}_{\tau + \gamma} |_{*}$. In fact, the rows/columns of $\mathrm{T}_{\Omega}$ differ from those of $\Omega^{(L-1)}_{\tau + \gamma} |_{*}$ only by overall multiplicative factors, which allow us to write
\< \label{R2}
\mathrm{det} (\mathrm{T}_{\Omega}) &=& [x_0 - \mu_L + L\gamma]^{L-1} [x_0 - \mu_L + \tau + (L+2)\gamma]^{L-1} \nonumber \\
&& \qquad \times \prod_{n=1}^{L-2} [x_0 - \mu_L + (n+1)\gamma]^n \; \mathrm{det} ( \left. \Omega^{(L-1)}_{\tau + \gamma} \right|_{*} ) \; .
\>
The combination of \eqref{R1} and \eqref{R2} then gives rise to a recurrence relation for the relevant determinant and, therefore, it suffices to show $\mathrm{det} ( \Omega^{(1)}_{\tau+\gamma}) \neq 0$ under the specializations $x_{\bar{0}} = x_0 + \tau + 2\gamma$ and $x_1 = x_0 + \gamma$ in order to complete our proof. The matrix $\Omega^{(1)}_{\tau+\gamma}$, in its turn, is a $2 \times 2$ matrix which becomes upper triangular under the aforementioned specializations. Hence, its determinant is simply given by the product of its diagonal entries and this is clearly non-vanishing. This concludes our proof.

\end{proof}

\subsection{Determinantal representations} \label{DET}

The analysis of the equation \eqref{eqADneu} has led us to the system of equations \eqref{eqADper}, from which \eqref{ZZ}
and \eqref{ZZZ} follows as a consequence. In particular, the structure of \eqref{ZZ} and \eqref{ZZZ} is quite appealing
from the perspective of expressing $Z_{\tau}$ as a determinant. The latter is our main goal and for that we first take a closer look at the structure of \eqref{ZZ} and \eqref{ZZZ}. For instance, the first equation in \eqref{ZZ} says $Z_{\tau} (X_i^0) \sim Z_{\tau} (X)$ and for such kind of relation
we can fix all variables in the set $\{x_1, \dots,  x_{i-1}, x_{i+1} , \dots , x_L \}$. The same argument also holds for the
second relation in \eqref{ZZ}. Thus each equation in \eqref{ZZ} is essentially an one-variable
functional relation which could be simply regarded as $Z_{\tau} (x_0) \sim Z_{\tau} (x_i)$ and $Z_{\tau} (x_{\bar{0}}) \sim Z_{\tau} (x_i)$.
Similarly, Eq. \eqref{ZZZ} can be regarded as a two-variable functional relation, i.e. $Z_{\tau} (x_0, x_{\bar{0}}) \sim Z_{\tau} (x_i , x_j)$.
This analysis paves the way for using \emph{separation of variables} for solving the aforementioned functional relations. 

The structure of \eqref{ZZ} and \eqref{ZZZ} has important consequences from the perspective of separation of variables. Those
relations tell us that the ratio of determinants defined by \eqref{FI}-\eqref{IKJbij} necessarily simplifies unveiling the solution
up to an overall multiplicative factor. More precisely, from \eqref{ZZ} and \eqref{ZZZ} we can infer that
\begin{align} \label{SEP}
\left| \begin{matrix} \mathcal{F} & \mathcal{I} & \mathcal{G} \cr \bar{\mathcal{I}} & \mathcal{K} & \mathcal{J} \cr \bar{\mathcal{F}} & \bar{\mathcal{J}} & \bar{\mathcal{G}} \end{matrix} \right| &= Z_{\tau} (X) \; f( x_0, x_{\bar{0}} , x_1 , \dots, x_L ) & \left| \begin{matrix} \mathcal{F}_i & \mathcal{I} & \mathcal{G} \cr \bar{\mathcal{I}}_i & \mathcal{K} & \mathcal{J} \cr \bar{\mathcal{F}}_i & \bar{\mathcal{J}} & \bar{\mathcal{G}} \end{matrix} \right| &= Z_{\tau} (X_i^{0}) \; f( x_0, x_{\bar{0}} , x_1 , \dots, x_L  ) \nonumber \\
\left| \begin{matrix} \mathcal{F} & \mathcal{I} & \mathcal{G}_i \cr \bar{\mathcal{I}} & \mathcal{K} & \mathcal{J}_i \cr \bar{\mathcal{F}} & \bar{\mathcal{J}} & \bar{\mathcal{G}}_i \end{matrix} \right| &= Z_{\tau} (X_i^{\bar{0}}) \; f( x_0, x_{\bar{0}} , x_1 , \dots, x_L ) & \left| \begin{matrix} \mathcal{F} & \mathcal{I}_{ij} & \mathcal{G} \cr \bar{\mathcal{I}} & \mathcal{K}_{ij} & \mathcal{J} \cr \bar{\mathcal{F}} & \bar{\mathcal{J}}_{ij} & \bar{\mathcal{G}} \end{matrix} \right| &= Z_{\tau} (X_{i,j}^{0, \bar{0}}) \; f( x_0, x_{\bar{0}} , x_1 , \dots, x_L )  \nonumber \\
\end{align}
for a given fixed function $f$. Thus our problem actually consists in extracting the partition function $Z_{\tau}$ out of those determinants.

Initially let us focus on the first relation of \eqref{ZZ} which expresses $Z_{\tau} (X_i^0) / Z_{\tau} (X) $
as the ratio of two determinants. From the inspection of \eqref{FG}-\eqref{FBI}, \eqref{l0}-\eqref{lm} and 
\eqref{coeffADneu} one can verify the matrix coefficients entering those determinants depend on the set of variables $X$ as well
as $x_0$ and $x_{\bar{0}}$. However, $Z_{\tau} (X_i^0) / Z_{\tau} (X) $ is independent of $x_{\bar{0}}$ and consequently
the corresponding ratio of determinants inherits this property. In this way, although the matrix entries \eqref{FG}-\eqref{FBI}
exhibit local dependence on $x_{\bar{0}}$, this variable has no global effect as far as the first relation of \eqref{ZZ} is concerned.
The second relation of \eqref{ZZ} and \eqref{ZZZ} allow one to draw similar conclusions involving the variable $x_0$.
We shall postpone the discussion on the role played by $x_0$ and $x_{\bar{0}}$; and proceed with the examination of the matrix
\< \label{omdef}
\Omega \coloneqq \begin{pmatrix} \mathcal{F} & \mathcal{I} & \mathcal{G} \cr \bar{\mathcal{I}} & \mathcal{K} & \mathcal{J} \cr \bar{\mathcal{F}} & \bar{\mathcal{J}} & \bar{\mathcal{G}} \end{pmatrix}  \; .
\>
We want to extract $Z_{\tau} (X)$ out of ${\rm det}(\Omega)$ and, taking into account \eqref{l0}-\eqref{lm} and \eqref{FG}-\eqref{KK},
we can see $\mathcal{F}$ and $\bar{\mathcal{G}}$ are full matrices whose diagonal entries consist of a sum of two products and off-diagonal entries given by a single product. In their turn $\bar{\mathcal{F}}$ and $\mathcal{G}$
are diagonal matrices with non-vanishing components given by single products. The matrices $\mathcal{I}$, $\bar{\mathcal{I}}$,
$\mathcal{J}$ and  $\bar{\mathcal{J}}$ are sparse and their non-null entries are single products. Lastly, $\mathcal{K}$ is also 
sparse but with diagonal components given by the sum of two products and off-diagonal ones consisting of single products.
The analysis of relations \eqref{SEP} then allow us to state the following theorem.

\begin{thm} \label{Z1}
The partition function $Z_{\tau} (X)$ can be written as 
\begin{eqnarray} \label{Z}
Z_{\tau} (X) &=& (-1)^L \left(\frac{[(L+1)\gamma]}{[\tau + (L+2)\gamma]}\right)^{d_{L-1}} \left(\frac{[\tau + (L+1)\gamma]}{[L\gamma]}\right)^{d_L} \prod_{i,j = 1}^L [x_i - \mu_j] \prod_{k=1}^L \frac{[k \gamma]}{[\tau + k \gamma]} \nonumber \\
&& \times \frac{[\sum_{l=1}^L (x_l - \mu_l) + (L+1)\gamma ]}{[\sum_{l=1}^L (x_l - \mu_l) + \tau + (L+2)\gamma ]} {\rm det} \left( \Omega \; \omega^{-1} \right) \; ,
\end{eqnarray}
where $\omega \coloneqq \left. \Omega \right|_{\tau = - \gamma}$. 
\end{thm}

In order to avoid an overcrowded section we present the proof of Theorem \ref{Z1} in \Appref{app:proofT1}.
As a matter of fact we can extract more results from the method used for proving Theorem \ref{Z1}. Once we determine the function 
$f(x_0 , x_{\bar{0}}, x_1 , \dots , x_L)$ appearing in \eqref{SEP}, we automatically find another three families of determinantal representations.
In order to precise the latter statement we introduce matrices
\< \label{BOIJ}
\Omega_i \coloneqq \begin{pmatrix} \mathcal{F}_i & \mathcal{I} & \mathcal{G} \cr \bar{\mathcal{I}}_i & \mathcal{K} & \mathcal{J} \cr \bar{\mathcal{F}}_i & \bar{\mathcal{J}} & \bar{\mathcal{G}} \end{pmatrix} \quad , \quad
\bar{\Omega}_i \coloneqq \begin{pmatrix} \mathcal{F} & \mathcal{I} & \mathcal{G}_i \cr \bar{\mathcal{I}} & \mathcal{K} & \mathcal{J}_i \cr \bar{\mathcal{F}} & \bar{\mathcal{J}} & \bar{\mathcal{G}}_i \end{pmatrix} \quad \text{and} \quad
\widetilde{\Omega}_{ij} \coloneqq \begin{pmatrix} \mathcal{F} & \mathcal{I}_{ij} & \mathcal{G} \cr \bar{\mathcal{I}} & \mathcal{K}_{ij} & \mathcal{J} \cr \bar{\mathcal{F}} & \bar{\mathcal{J}}_{ij} & \bar{\mathcal{G}} \end{pmatrix} \; . \nonumber \\
\>
The entries of $\Omega_i$, $\bar{\Omega}_i$ and $\widetilde{\Omega}_{ij}$ have been defined in \eqref{FI}-\eqref{IKJbij} and their structure are not too 
different from the ones constituting $\Omega$. In this way, taking into account \eqref{BOIJ}, we can state the following theorem.

\begin{thm} \label{Z234}
Write $\omega_i \coloneqq \left. \Omega_i \right|_{\tau = - \gamma}$ and $\bar{\omega}_i \coloneqq \left. \bar{\Omega}_i \right|_{\tau = - \gamma}$
for $1 \leq i \leq L$. Also, let us define  $\widetilde{\omega}_{i j} \coloneqq \left. \widetilde{\Omega}_{i j} \right|_{\tau = - \gamma}$ on the interval
$1 \leq i < j \leq L$. Then there exist families of continuous determinantal
representations for $Z_{\tau}$ given by
\begin{eqnarray} \label{Z0I}
Z_{\tau} (X_i^0) &=& (-1)^L \left(\frac{[(L+1)\gamma]}{[\tau + (L+2)\gamma]}\right)^{d_{L-1}} \left(\frac{[\tau + (L+1)\gamma]}{[L\gamma]}\right)^{d_L} \prod_{x \in X_i^0} \prod_{j = 1}^L [x - \mu_j] \prod_{k=1}^L \frac{[k \gamma]}{[\tau + k \gamma]} \nonumber \\
&& \times \frac{[\sum_{l=1}^L (x_l - \mu_l) + (L+1)\gamma ]}{[\sum_{l=1}^L (x_l - \mu_l) + \tau + (L+2)\gamma ]} {\rm det} \left( \Omega_i \; \omega_{i}^{-1} \right) \; ,
\end{eqnarray}
\begin{eqnarray} \label{Zb0I}
Z_{\tau} (X_i^{\bar{0}}) &=& (-1)^L \left(\frac{[(L+1)\gamma]}{[\tau + (L+2)\gamma]}\right)^{d_{L-1}} \left(\frac{[\tau + (L+1)\gamma]}{[L\gamma]}\right)^{d_L} \prod_{x \in X_i^{\bar{0}}} \prod_{j = 1}^L [x - \mu_j] \prod_{k=1}^L \frac{[k \gamma]}{[\tau + k \gamma]} \nonumber \\
&& \times \frac{[\sum_{l=1}^L (x_l - \mu_l) + (L+1)\gamma ]}{[\sum_{l=1}^L (x_l - \mu_l) + \tau + (L+2)\gamma ]} {\rm det} \left( \bar{\Omega}_i \; \bar{\omega}_{i}^{-1} \right) \; ,
\end{eqnarray}
\begin{eqnarray} \label{Z0b0IJ}
Z_{\tau} (X_{i,j}^{0, \bar{0}}) &=& (-1)^L \left(\frac{[(L+1)\gamma]}{[\tau + (L+2)\gamma]}\right)^{d_{L-1}} \left(\frac{[\tau + (L+1)\gamma]}{[L\gamma]}\right)^{d_L} \prod_{x \in X_{i, j}^{0, \bar{0}}} \prod_{k = 1}^L [x - \mu_k] \prod_{l=1}^L \frac{[l \gamma]}{[\tau + l \gamma]} \nonumber \\
&& \times \frac{[\sum_{m=1}^L (x_m - \mu_m) + (L+1)\gamma ]}{[\sum_{m=1}^L (x_m - \mu_m) + \tau + (L+2)\gamma ]} {\rm det} \left( \widetilde{\Omega}_{i, j} \; \widetilde{\omega}_{i,j}^{-1} \right) \; .
\end{eqnarray}
\end{thm}

\begin{rema}
Using $\Pi_{0,i} Z_{\tau} (X_i^0) = Z_{\tau} (X)$ one can compare formulae \eqref{Z} and \eqref{Z0I}. By doing so we find \eqref{Z} and \eqref{Z0I}
indeed consist of different representations. Similarly, one reaches the same conclusion for representations \eqref{Zb0I} and \eqref{Z0b0IJ}.
\end{rema}

\begin{rema}
Formulae \eqref{Z0I} and  \eqref{Zb0I} encloses $L$ representations each since $1 \leq i \leq L$. On the other hand, \eqref{Z0b0IJ} yields
$L(L-1)/2$ representations due to the condition $1 \leq i , j \leq L$. All together they add up to $L(L+3)/2$ families of representations.
\end{rema}

The proof of Theorem \ref{Z234} can also be found in \Appref{app:proofT1} and now let us focus again on the representation \eqref{Z}. The non-trivial part is contained in the matrix $\Omega$ \eqref{omdef} whose entries are
defined in \eqref{FG}-\eqref{KK}. In order to make our results more explicit we have also collected \eqref{FG}-\eqref{IKJbij} directly in terms
of theta-functions in \Appref{app:FUN}. From \eqref{FG}-\eqref{KK} one can see the entries of the matrices $\Omega$ and $\omega$ depends on the
variables $x_0$ and $x_{\bar{0}}$, in addition to $X$. However, the combination ${\rm det}(\Omega \; \omega^{-1})$ is independent of $x_0$ and $x_{\bar{0}}$,
and it essentially gives the partition function $Z_{\tau} (X)$. In this way, those two exceeding variables can be chosen at convenience without 
affecting the partition function. The same analysis holds for \eqref{Z0I}-\eqref{Z0b0IJ} where each representation also possess two extra local variables
having no global influence. Therefore, both Theorems \ref{Z1} and \ref{Z234} give us continuous families of single determinantal representations and one can regard
those extra variables as their parameterization. 

\begin{rema}
The representation recently presented in \cite{Galleas_pre_2016} consists of \eqref{Z} under specialization 
$x_0 = \mu_1 - 2\gamma$ and $x_{\bar{0}} = \mu_1 - \gamma$.
\end{rema}

\section{The six-vertex model limit} \label{sec:6V}

In this section we investigate a particular limit of the partition function \eqref{pf} where it reduces to that of the six-vertex model
with domain-wall boundaries introduced in \cite{Korepin_1982}. This is achieved in the limit $p \to 0$ followed by the limit $\tau \to \infty$. 
As far as the allowed lattice configurations are concerned, it was firstly pointed out by Lenard the existence of an equivalence between \textsc{sos}
and vertex models configurations. The specialization of the statistical weights requires a more careful analysis and for that we employ the identity
$\lim_{p \to 0} -\ii p^{- \frac{1}{4}} [x] = \sinh{(x)}$. Moreover, in this section we also consider the limit $\tau \to \infty$ in such a way that \eqref{bw} reduces to the standard six-vertex model
weights
\< \label{bw6v}
a(x) &\coloneqq& \sinh{(x+\gamma)} \nonumber \\
b(x) &\coloneqq& \sinh{(x)} \nonumber \\
c(x) &\coloneqq& \sinh{(\gamma)} \; .
\>
In writing \eqref{bw6v} we have ignored overall normalization factors. Also, as we are taking the limit $\tau \to \infty$, we  simply
use $Z_{\infty} (X) = Z(X)$ to denote the model's partition function. 

The structure of functional relations type A and D, given respectively by \eqref{eqA} and \eqref{eqD}, exhibit important simplifications
in the six-vertex model limit. This limit eliminates the \emph{dynamical} aspect of both equations, making the use of equation \eqref{eqADneu} unnecessary. Here we shall refer to \eqref{eqA} in the six-vertex model limit as \emph{reduced type A}
equation and, analogously, \eqref{eqD} will be \emph{reduced type D}. Moreover, both reduced type A and reduced type D equations are independently
able to fix the partition function $Z$ up to an overall multiplicative factor.
We shall tackle the resolution of these equations in what follows.

\subsection{Reduced type A} \label{sec:redA}

The derivation of Eq. \eqref{eqA} heavily relies on the dynamical Yang-Baxter algebra $\mathscr{A}(\mathcal{R})$. The latter reduces to the standard
Yang-Baxter algebra associated to the six-vertex model in the limit previously discussed. Therefore, we do not lose any relevant information by taking the 
six-vertex model limit directly on \eqref{eqA}. In fact, the functional relation \eqref{eqA} unfolds into a system of functional equations by means of permutations, namely
\[ \label{eqAl}
M_0^{(\mathcal{A},l)} \; Z_{\tau} (X_l^0) + \sum_{i=0}^L N_i^{(\mathcal{A},l)} \;  Z_{\tau + \gamma} (X_i^0) = 0 \; ,
\]
with coefficients
\<
M_0^{(\mathcal{A},l)} \coloneqq \Pi_{0, l} \;  M_0^{(\mathcal{A})}  \qquad \text{and} \qquad N_i^{(\mathcal{A},l)} \coloneqq \begin{cases}
\Pi_{0, l} \; N_l^{(\mathcal{A})}  \qquad i=0 \nonumber \\
\Pi_{0, l} \; N_0^{(\mathcal{A})}  \qquad i=l \nonumber \\
\Pi_{0, l} \; N_i^{(\mathcal{A})} \qquad \mbox{otherwise} \end{cases} \; . 
\>
The system of equations \eqref{eqAl} survives entirely in the six-vertex model limit and the latter feature can be recasted as the following corollary.

\begin{cor}
The partition function of the six-vertex model with domain-wall boundaries satisfies the system of equations
\[ \label{redAsys}
\sum_{i=0}^L \sigma_i^{(l)} \; Z(X_i^0) = 0 \qquad \qquad 0 \leq l \leq L \; ,
\]
with coefficients explicitly defined as
\<
\label{sigmaIL}
\sigma_i^{(l)} \coloneqq \begin{cases}
\displaystyle \frac{c(x_0 - x_l)}{b(x_0 - x_l)} \prod_{k=1}^L a(x_0 - \mu_k) \prod_{\substack{k=1 \\ k \neq l}}^L \frac{a(x_k - x_0)}{b(x_k - x_0)}  \qquad \qquad i = 0, \; l \neq 0 \nonumber \\
\displaystyle \prod_{k=1}^L b(x_l - \mu_k) - \prod_{k=1}^L a(x_l - \mu_k) \prod_{\substack{k=0 \\ k \neq l}}^L \frac{a(x_k - x_l)}{b(x_k - x_l)} \qquad i = l \nonumber \\
\displaystyle \frac{c(x_i - x_l)}{b(x_i - x_l)} \prod_{k=1}^L a(x_i - \mu_k)  \prod_{\substack{k=0 \\ k \neq i, l}}^L \frac{a(x_k - x_i)}{b(x_k - x_i)} \qquad \qquad \text{otherwise}
\end{cases} \; .
\>
\end{cor}
\begin{proof}
Straightforward evaluation of the limits $p \to 0$ and $\tau \to \infty$ in \eqref{eqAl}.
\end{proof}

\begin{rema}
The system \eqref{redAsys} encloses $L+1$ equations in consonance with the $L+1$ terms $Z(X_i^0)$ present in each equation. 
Moreover, one can also verify that ${\rm det} ( \sigma_i^{(l)} )_{0 \leq i, l \leq L } = 0$ ensuring
the system has a non-trivial solution.
\end{rema}

The resolution of \eqref{redAsys} can be performed along the lines described in \Secref{sec:PF}. For that we single out a subset of \eqref{redAsys}
containing $L$ equations, namely the ones on $1 \leq l \leq L$. This subset allows one to write each function $Z(X_i^0)$ for $1 \leq i \leq L$ in terms of 
$Z(X)$. By doing so we find
\[ \label{ZIZ}
Z(X_i^0) = \frac{{\rm det}(V_i)}{{\rm det}(V)} Z(X) 
\]
where $V$ and $V_i$ are $L \times L$ matrices with entries defined as follows,
\<
V_{\alpha, \beta} &\coloneqq& \sigma_{\beta}^{(\alpha)} \qquad \quad \quad \;\; 1 \leq \alpha, \beta \leq L \nonumber \\
( V_i ) _{\alpha, \beta} &\coloneqq& \begin{cases}
- \sigma_{0}^{(\alpha)} \qquad \quad \beta = i \nonumber \\
\sigma_{\beta}^{(\alpha)} \qquad \qquad \text{otherwise}
\end{cases} \; .
\>

As discussed in \Secref{DET}, one can regard Eq. \eqref{ZIZ} as an one-variable functional equation and employ separation of variables.
In this way we can conclude that ${\rm det}(V) = Z(X) f(x_0, x_1, \dots , x_L)$ and ${\rm det}(V_i) = Z(X_i^0) f(x_0, x_1, \dots , x_L)$ for a
given function $f$. Hence the full characterization of the partition function $Z$ only requires the determination of the function $f$.

\begin{thm} \label{DET6VA}
The partition function $Z$ can be written as
\[ \label{det6vA}
Z(X) = {\rm det}(V) \prod_{k=1}^L \frac{b(x_k - x_0)}{a(x_0 - \mu_k)}  \qquad \text{and} \qquad Z(X_i^0) = {\rm det}(V_i) \prod_{k=1}^L \frac{b(x_k - x_0)}{a(x_0 - \mu_k)} 
\]
where
\<
V_{\alpha, \beta} &=& \begin{cases}
\displaystyle \prod_{k=1}^L b(x_{\alpha} - \mu_k) - \prod_{k=1}^L a(x_{\alpha} - \mu_k) \prod_{\substack{k=0 \\ k \neq \alpha}}^L \frac{a(x_k - x_{\alpha})}{b(x_k - x_{\alpha})} \qquad \;\; \beta = \alpha \nonumber \\
\displaystyle \frac{c(x_{\beta} - x_{\alpha})}{b(x_{\beta} - x_{\alpha})} \prod_{k=1}^L a(x_{\beta} - \mu_k)  \prod_{\substack{k=0 \\ k \neq \alpha, \beta}}^L \frac{a(x_k - x_{\beta})}{b(x_k - x_{\beta})} \qquad \qquad \text{otherwise}
\end{cases} \; , \nonumber \\
( V_i )_{\alpha, \beta} &=&  \begin{cases}
\displaystyle \frac{c(x_{\alpha} - x_0)}{b(x_{\alpha} - x_0)} \prod_{k=1}^L a(x_0 - \mu_k) \prod_{\substack{k=1 \\ k \neq \alpha}}^L \frac{a(x_k - x_0)}{b(x_k - x_0)}  \qquad \qquad \quad \beta = i, \; 1 \leq \alpha \leq L \nonumber \\
\displaystyle \prod_{k=1}^L b(x_{\alpha} - \mu_k) - \prod_{k=1}^L a(x_{\alpha} - \mu_k) \prod_{\substack{k=0 \\ k \neq \alpha}}^L \frac{a(x_k - x_{\alpha})}{b(x_k - x_{\alpha})} \qquad \; \beta = \alpha \neq i \nonumber \\
\displaystyle \frac{c(x_{\beta} - x_{\alpha})}{b(x_{\beta} - x_{\alpha})} \prod_{k=1}^L a(x_{\beta} - \mu_k)  \prod_{\substack{k=0 \\ k \neq \alpha, \beta}}^L \frac{a(x_k - x_{\beta})}{b(x_k - x_{\beta})} \qquad \qquad \text{otherwise}
\end{cases} \; .
\>
\end{thm}

The proof of Theorem \ref{DET6VA} can be found in \Appref{app:proofAD} and some comments concerning formulae \eqref{det6vA} are important at this stage.
For instance, the second expression of \eqref{det6vA} can be matched with the first one upon the identification $\Pi_{0, i} Z(X_i^0) = Z(X)$.
However, one can notice $\Pi_{0, i} V_i$ does not seem to be related to $V$ by simple transformations. This indicates they indeed constitute independent representations.
In addition to that, one of the most important aspects of both formulae in \eqref{det6vA} is that they consist of continuous families of representations. 
For example, let us take a closer look at the first expression of \eqref{det6vA}. The LHS depends on variables $X$ while the matrix entries in the RHS
depends on $X \cup \{ x_0 \}$. Thus the dependence of the RHS with $x_0$ is local but not global; and $x_0$ can be regarded as a variable parameterizing
this continuous family of representations. This feature of representations \eqref{det6vA} might have important consequences as far as applications are
concerned. For instance, depending on the application we have in mind, the variable $x_0$ can be suitably tuned in order to simplify calculations.
This same argument holds for the second formula of \eqref{det6vA} but now with variable $x_i$ parameterizing the family of representations.

\subsection{Reduced type D} \label{sec:redD}

We continue our analysis of the six-vertex model limit along the same lines described in \Secref{sec:redA}. For that we firstly
extend \eqref{eqD} to a system of functional relations through permutations $\Pi_{0, m}$ along the lines used in \eqref{eqAl}. 
Here we shall not present explicitly the resulting system of equations since we are only interested in the six-vertex model limit. The latter is
then formalized as the following corollary.

\begin{cor} The system of equations
\[ \label{redDsys}
\sum_{i=0}^L \rho_i^{(m)} \; Z(X_i^0) = 0 \qquad \qquad 0 \leq m \leq L \; ,
\]
with coefficients 
\<
\label{sigmaIL}
\rho_i^{(m)} \coloneqq \begin{cases}
\displaystyle \frac{c(x_m - x_0)}{b(x_m - x_0)} \prod_{k=1}^L b(x_0 - \mu_k) \prod_{\substack{k=1 \\ k \neq m}}^L \frac{a(x_0 - x_k)}{b(x_0 - x_k)}  \qquad \qquad \qquad i = 0, \; m \neq 0 \nonumber \\
\displaystyle \prod_{k=1}^L a(x_m - \mu_k) - \prod_{k=1}^L b(x_m - \mu_k) \prod_{\substack{k=0 \\ k \neq m}}^L \frac{a(x_m - x_k)}{b(x_m - x_k)} \qquad \quad i = m \nonumber \\
\displaystyle \frac{c(x_m - x_i)}{b(x_m - x_i)} \prod_{k=1}^L b(x_i - \mu_k)  \prod_{\substack{k=0 \\ k \neq i, m}}^L \frac{a(x_i - x_k)}{b(x_i - x_k)} \qquad \qquad \qquad \text{otherwise}
\end{cases} \; .
\>
is satisfied by the partition function of the six-vertex model with domain-wall boundaries.
\end{cor}

\begin{proof}
We firstly apply $\Pi_{0, m}$ onto \eqref{eqD} and then take the limits $p \to 0$ and $\tau \to \infty$.
\end{proof}

\begin{rema}
The label $(m)$ in \eqref{redDsys} takes values on the interval $0 \leq m \leq L$ and thus \eqref{redDsys} consists of a system of $L+1$ equations.
Also, \eqref{redDsys} relates $L+1$ unknowns of the form $Z(X_i^0)$, and the existence of non-trivial solutions requires 
${\rm det}( \rho_i^{(m)} )_{0 \leq i, m \leq L} = 0$. The latter condition can be readily verified.
\end{rema}

\begin{rema}
In contrast to \eqref{eqD}, we have used $x_0$ instead of $x_{\bar{0}}$ since this distinction is not required for the analysis of 
\eqref{redDsys}.
\end{rema}

Next we would like to solve the system of equations \eqref{redDsys} and this can be accomplished using the methodology
described in \Secref{sec:PF} and \Secref{sec:redA}. For that we focus on the subset of \eqref{redDsys} formed by 
$1 \leq m \leq L$ which allows one to express each function $Z(X_i^0)$ in terms of $Z(X)$. This possibility
is a direct consequence of the linearity of our equations. Using Cramer's rule we then find relations of the form
\[ \label{ZIZd}
Z(X_i^0) = \frac{{\rm det}(W_i)}{{\rm det}(W)} Z(X) \; ,
\]
The matrices $W$ and $W_i$ in \eqref{ZIZd} have dimension $L \times L$ and their entries are given in terms of the coefficients
of \eqref{redDsys}. More precisely, they are defined as
\<
W_{\alpha, \beta} &\coloneqq& \rho_{\beta}^{(\alpha)} \qquad \quad \quad \;\; 1 \leq \alpha, \beta \leq L \nonumber \\
( W_i ) _{\alpha, \beta} &\coloneqq& \begin{cases}
- \rho_{0}^{(\alpha)} \qquad \quad \beta = i \nonumber \\
\rho_{\beta}^{(\alpha)} \qquad \qquad \text{otherwise}
\end{cases} \; .
\>
As previously discussed, one can regard \eqref{ZIZd} as a system of one-variable functional equations which can be solved using
separation of variables. The resolution of such equations produces the following theorem.

\begin{thm} \label{DET6VD}
The partition function $Z$ admits the following families of continuous representations
\[ \label{det6vD}
Z(X) = {\rm det}(W) \prod_{k=1}^L \frac{b(x_0 - x_k)}{b(x_0 - \mu_k)}  \qquad \text{and} \qquad Z(X_i^0) = {\rm det}(W_i) \prod_{k=1}^L \frac{b(x_0 - x_k)}{b(x_0 - \mu_k)} 
\]
where
\<
W_{\alpha, \beta} &=& \begin{cases}
\displaystyle \prod_{k=1}^L a(x_{\alpha} - \mu_k) - \prod_{k=1}^L b(x_{\alpha} - \mu_k) \prod_{\substack{k=0 \\ k \neq \alpha}}^L \frac{a(x_{\alpha} - x_k)}{b(x_{\alpha} - x_k)} \qquad \;\; \beta = \alpha \nonumber \\
\displaystyle \frac{c(x_{\alpha} - x_{\beta})}{b(x_{\alpha} - x_{\beta})} \prod_{k=1}^L b(x_{\beta} - \mu_k)  \prod_{\substack{k=0 \\ k \neq \alpha, \beta}}^L \frac{a(x_{\beta} - x_k)}{b(x_{\beta} - x_k)} \qquad \qquad \text{otherwise}
\end{cases} \; , \nonumber \\
( W_i )_{\alpha, \beta} &=&  \begin{cases}
\displaystyle \frac{c(x_0 - x_{\alpha} )}{b(x_0 - x_{\alpha})} \prod_{k=1}^L b(x_0 - \mu_k) \prod_{\substack{k=1 \\ k \neq \alpha}}^L \frac{a(x_0 - x_k)}{b(x_0 - x_k)}  \qquad \qquad \quad \beta = i, \; 1 \leq \alpha \leq L \nonumber \\
\displaystyle \prod_{k=1}^L a(x_{\alpha} - \mu_k) - \prod_{k=1}^L b(x_{\alpha} - \mu_k) \prod_{\substack{k=0 \\ k \neq \alpha}}^L \frac{a(x_{\alpha} - x_k)}{b(x_{\alpha} - x_k)} \qquad \; \beta = \alpha \neq i \nonumber \\
\displaystyle \frac{c(x_{\alpha} - x_{\beta})}{b(x_{\alpha} - x_{\beta})} \prod_{k=1}^L b(x_{\beta} - \mu_k)  \prod_{\substack{k=0 \\ k \neq \alpha, \beta}}^L \frac{a(x_{\beta} - x_k)}{b(x_{\beta} - x_k)} \qquad \qquad \text{otherwise}
\end{cases} \; .
\>
\end{thm}

The proof of Theorem \ref{DET6VD} is also given in \Appref{app:proofAD}. It is important to remark here that formulae \eqref{det6vD} consist of $L+1$
independent families of continuous representations with properties similar to the ones pointed out for \eqref{det6vA} in \Secref{sec:redA}. 
Although we shall not discuss those properties again in this section, we remark that altogether formulae \eqref{det6vA} and \eqref{det6vD} totals 
$2(L+1)$ families of continuous representations for the partition function $Z$.

\appendix

\section{Proofs of Theorems \ref{Z1} and \ref{Z234}} \label{app:proofT1}

The proof of Theorems \ref{Z1} and \ref{Z234} follows from the determination of the function $f$ appearing in \eqref{SEP}. For that we firstly notice that
relations \eqref{SEP} imply the following system of partial differential equations,
\begin{align} \label{PDEs}
\frac{\partial }{\partial x_0} \left( \frac{f}{{\rm det}(\Omega)}   \right) &= 0 &  \frac{\partial }{\partial x_{\bar{0}}} \left( \frac{f}{{\rm det}(\Omega)}   \right) &= 0  \nonumber \\
\frac{\partial }{\partial x_i} \left( \frac{f}{{\rm det}(\Omega_i)}   \right) &= 0 & \frac{\partial }{\partial x_{\bar{0}}} \left( \frac{f}{{\rm det}(\Omega_i)}   \right) &= 0 \nonumber  \\
\frac{\partial }{\partial x_i} \left( \frac{f}{{\rm det}(\bar{\Omega}_i)}   \right) &= 0 & \frac{\partial }{\partial x_{0}} \left( \frac{f}{{\rm det}(\bar{\Omega}_i)}   \right) &= 0  \nonumber \\
\frac{\partial }{\partial x_i} \left( \frac{f}{{\rm det}(\widetilde{\Omega}_{ij})}   \right) &= 0 & \frac{\partial }{\partial x_j} \left( \frac{f}{{\rm det}(\widetilde{\Omega}_{ij})}   \right) &= 0 \; .
\end{align}
Next we consider \eqref{ZZ} and \eqref{ZZZ} under a particular specialization of the dynamical parameter $\tau$. More precisely, for 
$\tau = - \gamma$ we find
\< \label{taugam}
\left. \frac{{\rm det}(\Omega_i)}{{\rm det}(\Omega)} \right|_{\tau = - \gamma} &=& \left. \frac{Z_{\tau} (X_i^0)}{Z_{\tau} (X)} \right|_{\tau = -\gamma } = \prod_{k=1}^L \frac{[x_0 - \mu_k]}{[x_i - \mu_k]} \nonumber \\
\left. \frac{{\rm det}(\bar{\Omega}_i)}{{\rm det}(\Omega)} \right|_{\tau = - \gamma} &=& \left. \frac{Z_{\tau} (X_i^{\bar{0}})}{Z_{\tau} (X)} \right|_{\tau = -\gamma }  = \prod_{k=1}^L \frac{[x_{\bar{0}} - \mu_k]}{[x_i - \mu_k]} \nonumber \\
\left. \frac{{\rm det}(\widetilde{\Omega}_{ij})}{{\rm det}(\Omega)} \right|_{\tau = - \gamma} &=& \left. \frac{Z_{\tau} (X_{i,j}^{0,\bar{0}})}{Z_{\tau} (X)} \right|_{\tau = -\gamma }= \prod_{k=1}^L \frac{[x_0 - \mu_k][x_{\bar{0}} - \mu_k]}{[x_i - \mu_k][x_j - \mu_k]} \; .
\>
Equations \eqref{taugam} can now be easily solved yielding
\[ \label{soltg}
\left. Z_{\tau} (X) \right|_{\tau = -\gamma } = \mathcal{C}_1 \prod_{i,j=1}^L [x_i - \mu_j]
\]
with $\mathcal{C}_1$ being an $X$-independent term. Now we look at \eqref{SEP} under the same specialization taking into account solution \eqref{soltg}.
This gives us the following relations:
\begin{align} \label{ftg}
\left. f \right|_{\tau = - \gamma} & = \frac{\mathcal{C}_1^{-1} \left. {\rm det} (\Omega)  \right|_{\tau = - \gamma}}{\displaystyle \prod_{x \in X} \prod_{k=1}^L [x - \mu_k]} & =  & \frac{\mathcal{C}_1^{-1} \left. {\rm det} (\Omega_i)  \right|_{\tau = - \gamma}}{\displaystyle \prod_{x \in X_i^0} \prod_{k=1}^L [x - \mu_k]} \nonumber \\
& = \frac{\mathcal{C}_1^{-1} \left. {\rm det} (\bar{\Omega}_i)  \right|_{\tau = - \gamma}}{\displaystyle \prod_{x \in X_i^{\bar{0}}} \prod_{k=1}^L [x - \mu_k]} & = & \frac{\mathcal{C}_1^{-1} \left. {\rm det} (\widetilde{\Omega}_{ij})  \right|_{\tau = - \gamma}}{\displaystyle \prod_{x \in X_{i,j}^{0, \bar{0}}} \prod_{k=1}^L [x - \mu_k]} \; .
\end{align}
Motivated by \eqref{ftg} we also consider the redefinition 
$$f( x_0, x_{\bar{0}} , x_1 , \dots, x_L  ) \eqqcolon \left. f \right|_{\tau = - \gamma} \; f_{r} ( x_0, x_{\bar{0}} , x_1 , \dots, x_L  ) \; . $$
In addition to that we also introduce matrices $T$, $T_i$, $\bar{T}_i$ and $\widetilde{T}_{ij}$ such that
\begin{align}
{\rm det} (T) & = \prod_{x \in X} \prod_{k=1}^L [x - \mu_k] \frac{{\rm det}(\Omega)}{\left. {\rm det}(\Omega) \right|_{\tau = - \gamma}} & {\rm det} (T_i) & = \prod_{x \in X_i^0} \prod_{k=1}^L [x - \mu_k] \frac{{\rm det}(\Omega_i)}{\left. {\rm det}(\Omega_i) \right|_{\tau = - \gamma}} \nonumber \\
{\rm det} (\bar{T}_i) & = \prod_{x \in X_i^{\bar{0}}} \prod_{k=1}^L [x - \mu_k] \frac{{\rm det}(\bar{\Omega}_i)}{\left. {\rm det}(\bar{\Omega}_i) \right|_{\tau = - \gamma}} & {\rm det} (\widetilde{T}_{ij}) & = \prod_{x \in X_{i,j}^{0, \bar{0}}} \prod_{k=1}^L [x - \mu_k] \frac{{\rm det}(\widetilde{\Omega}_{ij})}{\left. {\rm det}(\widetilde{\Omega}_{ij}) \right|_{\tau = - \gamma}} \; . \nonumber \\
\end{align}
Hence, we can rewrite the system of equations \eqref{PDEs} as
\begin{multicols}{2}
\<
\frac{\partial }{\partial x_0} \left( \frac{f_r}{{\rm det}(T)}   \right) &=& 0 \label{t1} \\
\frac{\partial }{\partial x_i} \left( \frac{f_r}{{\rm det}(T_i)}   \right) &=& 0 \label{t3} \\
\frac{\partial }{\partial x_i} \left( \frac{f_r}{{\rm det}(\bar{T}_i)}   \right) &=& 0 \label{t5} \\
\frac{\partial }{\partial x_i} \left( \frac{f_r}{{\rm det}(\widetilde{T}_{ij})}   \right) &=& 0 \label{t7}
\>

\<
\frac{\partial }{\partial x_{\bar{0}}} \left( \frac{f_r}{{\rm det}(T)}   \right) &=& 0 \label{t2} \\
\frac{\partial }{\partial x_{\bar{0}}} \left( \frac{f_r}{{\rm det}(T_i)}   \right) &=& 0 \label{t4} \\
\frac{\partial }{\partial x_{0}} \left( \frac{f_r}{{\rm det}(\bar{T}_i)}   \right) &=& 0 \label{t6} \\
\frac{\partial }{\partial x_j} \left( \frac{f_r}{{\rm det}(\widetilde{T}_{ij})}   \right) &=& 0 \; . \label{t8}
\>
\end{multicols}
Equations \eqref{t1}, \eqref{t2}, \eqref{t4} and \eqref{t6} simplifies to 
\[
\frac{\partial \log{(f_r)}}{\partial x_0} =  \frac{\partial \log{(f_r)}}{\partial x_{\bar{0}}} = 0 \; .
\]
Therefore we can conclude that $f_r ( x_0, x_{\bar{0}} , x_1 , \dots, x_L  ) = f_r ( x_1, \dots , x_L  )$. On the other hand, Eqs. \eqref{t3}, \eqref{t5}, \eqref{t7} and \eqref{t8} gives
\<
\frac{\partial \log{(f_r)}}{\partial x_i} &=& \frac{\partial}{\partial x_i} \log{([ \sum_{k=1}^L (x_k - \mu_k) + \tau + (L+2)\gamma ])} \nonumber \\
&& - \frac{\partial}{\partial x_i} \log{([ \sum_{k=1}^L (x_k - \mu_k) + (L+1)\gamma ])} \; ,
\>
which can be readily integrated. In this way we find
\[ \label{fran}
f_r (X) = \mathcal{C}_0 \frac{[ \sum_{k=1}^L (x_k - \mu_k) + \tau + (L+2)\gamma ]}{[\sum_{k=1}^L (x_k - \mu_k) + (L+1)\gamma ]} \; .
\]
The combination of \eqref{SEP}, \eqref{ftg} and \eqref{fran} leaves us with the following expressions,
\<
Z_{\tau} (X) &=& \frac{\mathcal{C}_1}{\mathcal{C}_0} \frac{[\sum_{k=1}^L (x_k - \mu_k) + (L+1)\gamma ]}{[ \sum_{k=1}^L (x_k - \mu_k) + \tau + (L+2)\gamma ]} \frac{{\rm det} (\Omega)}{\left. {\rm det} (\Omega) \right|_{\tau = - \gamma}} \prod_{x \in X} \prod_{k=1}^L [x - u_k] \nonumber \\
Z_{\tau} (X_i^0) &=& \frac{\mathcal{C}_1}{\mathcal{C}_0} \frac{[\sum_{k=1}^L (x_k - \mu_k) + (L+1)\gamma ]}{[ \sum_{k=1}^L (x_k - \mu_k) + \tau + (L+2)\gamma ]} \frac{{\rm det} (\Omega_i)}{\left. {\rm det} (\Omega_i) \right|_{\tau = - \gamma}} \prod_{x \in X_i^0} \prod_{k=1}^L [x - u_k] \nonumber \\
Z_{\tau} (X_i^{\bar{0}}) &=& \frac{\mathcal{C}_1}{\mathcal{C}_0} \frac{[\sum_{k=1}^L (x_k - \mu_k) + (L+1)\gamma ]}{[ \sum_{k=1}^L (x_k - \mu_k) + \tau + (L+2)\gamma ]} \frac{{\rm det} (\bar{\Omega}_i)}{\left. {\rm det} (\bar{\Omega}_i) \right|_{\tau = - \gamma}} \prod_{x \in X_i^{\bar{0}}} \prod_{k=1}^L [x - u_k] \nonumber \\
Z_{\tau} (X_{i,j}^{0, \bar{0}}) &=& \frac{\mathcal{C}_1}{\mathcal{C}_0} \frac{[\sum_{k=1}^L (x_k - \mu_k) + (L+1)\gamma ]}{[ \sum_{k=1}^L (x_k - \mu_k) + \tau + (L+2)\gamma ]} \frac{{\rm det} (\widetilde{\Omega}_{ij})}{\left. {\rm det} (\widetilde{\Omega}_{ij}) \right|_{\tau = - \gamma}} \prod_{x \in X_{i,j}^{0, \bar{0}}} \prod_{k=1}^L [x - u_k] \; , \nonumber \\
\>
and to complete the proofs of Theorems \ref{Z1} and \ref{Z234} we only need to determine the constant factor $\mathcal{C}_1 / \mathcal{C}_0$.
The latter can be obtained from the asymptotic behavior derived in \cite{Galleas_2013}. By doing so we find
\[
\frac{\mathcal{C}_1}{\mathcal{C}_0} = (-1)^L \left(\frac{[(L+1)\gamma]}{[\tau + (L+2)\gamma]}\right)^{d_{L-1}} \left(\frac{[\tau + (L+1)\gamma]}{[L\gamma]}\right)^{d_L} \prod_{k=1}^L \frac{[k \gamma]}{[\tau + k \gamma]} \; .
\]

\section{Proofs of Theorems \ref{DET6VA} and \ref{DET6VD}} \label{app:proofAD}

The proof of Theorem \ref{DET6VA} can be obtained from the analysis of relation \eqref{ZIZ}. This relation allows one to conclude that
\< \label{VVI}
{\rm det}(V_i) &=& Z(X_i^0) f(x_0, x_1, \dots , x_L) \nonumber \\
{\rm det}(V) &=& Z(X) f(x_0, x_1, \dots , x_L) \; ,
\>
and formulae \eqref{det6vA} follows from the determination of the function $f$. In order to determine such function we notice that
decomposition \eqref{VVI} induces a system of partial differential equations, namely
\<
\frac{\partial }{\partial x_i} \left( \frac{f}{{\rm det}(V_i)} \right) &=& 0  \label{D1} \\
\frac{\partial }{\partial x_0} \left( \frac{f}{{\rm det}(V)} \right) &=& 0 \label{D2} \; ,
\>
which can be solved for $f$. The solution of \eqref{D1} and \eqref{D2} can be obtained using elementary
methods and we start by noticing that \eqref{D1} simplifies to
\< \label{eqlogf}
\frac{\partial \log{(f)}}{\partial x_i} = \frac{1}{\tanh{(x_0 - x_i)}} \; .
\>
The solution of \eqref{eqlogf} is then  given by
\[ \label{logf}
\log{(f)} = - \log{\left( \prod_{k=1}^L b(x_0 - x_k) \right)} + h(x_0) \; ,
\]
where $h$ is an unknown function depending on the single variable $x_0$. Next we substitute \eqref{logf} in \eqref{D2} 
and this procedure yields the following constraint on $h$,
\< \label{h0}
\frac{\partial h}{\partial x_0} = \sum_{k=1}^{L} \frac{1}{\tanh{(x_0 - \mu_k + \gamma)}} \; .
\>
Eq. \eqref{h0} can be easily integrated and we find
\[ \label{hh0}
h(x_0) = \log{\left( \prod_{k=1}^L a(x_0 - \mu_k) \right)} + \log{(\mathcal{C})} 
\]
where $\mathcal{C}$ is an integration constant. Gathering \eqref{logf} and \eqref{hh0} we obtain
\[
f = \mathcal{C} \prod_{k=1}^{L} \frac{a(x_0 - \mu_k)}{b(x_0 - x_k)} \; ,
\]	
thus reducing our problem to the determination of $\mathcal{C}$. Using the asymptotic behavior derived in \cite{Galleas_2010} we obtain $\mathcal{C} = (-1)^L$ which completes the proof of Theorem \ref{DET6VA}. 

Next we detail the proof of Theorem \ref{DET6VD}. It is analogous to the one presented for Theorem \ref{DET6VA} and we start our analysis
from relation \eqref{ZIZd}. The latter allows one to write
\< \label{VVId}
{\rm det}(W_i) &=& Z(X_i^0) \bar{f}(x_0, x_1, \dots , x_L) \nonumber \\
{\rm det}(W) &=& Z(X) \bar{f}(x_0, x_1, \dots , x_L) \; ,
\>
and our task becomes the determination of the function $\bar{f}$. For that we consider the following system of partial differential equations,
\<
\frac{\partial }{\partial x_i} \left( \frac{\bar{f}}{{\rm det}(W_i)} \right) &=& 0  \label{DD1} \\
\frac{\partial }{\partial x_0} \left( \frac{\bar{f}}{{\rm det}(W)} \right) &=& 0 \label{DD2} \; ,
\>
which is a direct consequence of \eqref{VVId}. The resolution of \eqref{DD1} yields
\[ \label{logfb}
\log{(\bar{f})} = - \log{\left( \prod_{k=1}^L b(x_0 - x_k) \right)} + \bar{h}(x_0) \; ,
\]
where $\bar{h}$ is an arbitrary function depending solely on $x_0$. The substitution of \eqref{logfb} in \eqref{DD2}
yields a constraint for the function $\bar{h}$. This constraint can be readily solved and we find
\[ \label{hh0b}
\bar{h}(x_0) = \log{\left( \prod_{k=1}^L b(x_0 - \mu_k) \right)} + \log{(\bar{\mathcal{C}})} \; ,
\]
where $\bar{\mathcal{C}}$ is an integration constant. The combination of \eqref{logfb} and \eqref{hh0b} leaves us with the expression
\[
\bar{f} = \bar{\mathcal{C}} \prod_{k=1}^{L} \frac{b(x_0 - \mu_k)}{b(x_0 - x_k)} \; .
\]
Lastly, we need to determine the constant $\bar{\mathcal{C}}$ and from the asymptotic behavior presented in \cite{Galleas_2010} we find 
$\bar{\mathcal{C}} = 1$ which completes our proof.

\section{Explicit formulae} \label{app:FUN}

In this appendix we collect explicit expressions for the matrices \eqref{FG}-\eqref{IKJbij} in terms of the theta-function defined in \eqref{theta}.
These expressions are required for constructing the determinantal representations \eqref{Z} and \eqref{Z0I}-\eqref{Z0b0IJ}.
For convenience, we also introduce the generalized notation
\[
X_{a_1, a_2, \dots , a_n}^{b_1, b_2, \dots, b_m} \coloneqq X \cup \{ x_{b_1} , x_{b_2} , \dots , x_{b_m} \} \backslash \{ x_{a_1} , x_{a_2} , \dots , x_{a_n} \} \; .
\]

As far as representation \eqref{Z} is concerned, we need to build the matrix $\Omega$ defined in \eqref{omdef}. Its entries explicitly read
\<
&& \mathcal{F}_{a,b} = \nonumber \\
&& \begin{cases}
\displaystyle \frac{[x_a - x_0 + \gamma]}{[x_a - x_0]} \prod_{x \in X_a} \frac{[x_a - x + \gamma]}{[x_a - x]} \prod_{j=1}^L \frac{ [x_a - \mu_j] [x_{\bar{0}} - \mu_j + \gamma]}{ [x_a - \mu_j + \gamma]} \nonumber \\
\displaystyle \qquad \qquad \qquad - \; \frac{[x_{\bar{0}} - x_0 + \gamma] }{[x_{\bar{0}} - x_0 ]} \prod_{x \in X_a} \frac{[x_{\bar{0}} - x + \gamma] }{[x_{\bar{0}} - x ]} \prod_{j=1}^L [x_{\bar{0}} - \mu_j] \hfill  \;\; a = b \nonumber \\
\displaystyle - \frac{[\gamma] [x_a - x_b + \tau + (L+1)\gamma]}{[\tau + (L+1)\gamma] [x_a - x_b]} \prod_{j=1}^L \frac{[x_b - \mu_j] [x_{\bar{0}} - \mu_j + \gamma]}{[x_a - \mu_j + \gamma]}  \prod_{x \in X_{a,b}^0} \frac{[x_b - x + \gamma]}{[x_b - x]} \qquad \text{otherwise}
\end{cases} \\
\>
\<
&& \bar{\mathcal{F}}_{a,b} = \nonumber \\
&&\begin{cases}
\displaystyle - \frac{[\gamma] [x_0 - x_{\bar{0}} + \tau + (L+1)\gamma]}{[\tau + (L+1)\gamma] [x_0 - x_{\bar{0}}]} \prod_{j=1}^L \frac{[x_{\bar{0}} - \mu_j] [x_a - \mu_j + \gamma]}{[x_0 - \mu_j + \gamma]} \prod_{x \in X_a} \frac{[x_{\bar{0}} - x + \gamma]}{[x_{\bar{0}} - x]} \qquad \quad \;\; a = b \nonumber \\
0 \hfill \text{otherwise}
\end{cases} \\ 
\>
\<
&& \mathcal{G}_{a,b} = \nonumber \\
&&\begin{cases}
\displaystyle \frac{[\gamma] [x_{\bar{0}} - x_0 + \tau + (L+1)\gamma]}{[\tau + (L+1)\gamma] [x_{\bar{0}} - x_0]} \prod_{j=1}^L [x_0 - \mu_j] \prod_{x \in X_a} \frac{[x_0 - x + \gamma]}{[x_0 - x]} \qquad \qquad \qquad \quad \; \qquad \;\; a = b \nonumber \\
0 \hfill \text{otherwise}
\end{cases} \\
\>
\<
&& \bar{\mathcal{G}}_{a,b} = \nonumber \\
&&\begin{cases}
\displaystyle \frac{[x_0 - x_{\bar{0}} + \gamma]}{[x_0 - x_{\bar{0}}]} \prod_{x \in X_a} \frac{[x_0 - x + \gamma]}{[x_0 - x]} \prod_{j=1}^L \frac{[x_0 - \mu_j] [x_a - \mu_j + \gamma]}{[x_0 - \mu_j + \gamma]} \nonumber \\
\displaystyle \qquad \qquad \qquad - \; \frac{[x_a - x_{\bar{0}} + \gamma]}{[x_a - x_{\bar{0}}]} \prod_{x \in X_a} \frac{[x_a - x + \gamma]}{[x_a - x]} \prod_{j=1}^L [x_a - \mu_j] \qquad \qquad \qquad \qquad \;\; a = b \nonumber \\
\displaystyle \frac{[\gamma] [x_a - x_b + \tau + (L+1)\gamma]}{[\tau + (L+1)\gamma] [x_a - x_b]} \prod_{j=1}^L [x_b - \mu_j] \prod_{x \in X_{a,b}^{\bar{0}}}  \frac{[x_b - x + \gamma]}{[x_b - x]} \hfill \text{otherwise}
\end{cases} \\
\>
\<
&& \mathcal{I}_{a , n_{r,s}} = \nonumber \\
&& \begin{cases}
\displaystyle \frac{[\gamma] [x_{\bar{0}} - x_r + \tau + (L+1)\gamma]}{[\tau + (L+1)\gamma] [x_{\bar{0}} - x_r]} \prod_{j=1}^L [x_r - \mu_j] \prod_{x \in X_{r,s}^0} \frac{[x_r - x + \gamma]}{[x_r - x]} \qquad \qquad \qquad \qquad a = s \nonumber \\
\displaystyle \frac{[\gamma] [x_{\bar{0}} - x_s + \tau + (L+1)\gamma]}{[\tau + (L+1)\gamma] [x_{\bar{0}} - x_s]} \prod_{j=1}^L [x_s - \mu_j] \prod_{x \in X_{r,s}^0} \frac{[x_s - x + \gamma]}{[x_s - x]} \hfill a = r \nonumber \\
0 \hfill \text{otherwise}
\end{cases} \\
\>
\<
&& \bar{\mathcal{J}}_{a , n_{r,s}} = \nonumber \\
&& \begin{cases}
\displaystyle - \frac{[\gamma] [x_0 - x_r  + \tau + (L+1)\gamma]}{[\tau + (L+1)\gamma] [x_0 - x_r]} \prod_{j=1}^L \frac{[x_r - \mu_j] [x_a - \mu_j + \gamma]}{[x_0 - \mu_j + \gamma]} \nonumber \\
\displaystyle \qquad\qquad\qquad\qquad\qquad\qquad\qquad \times \prod_{x \in X_{r, s}^{\bar{0}}} \frac{[x_r - x + \gamma]}{[x_r - x]} \qquad \qquad \qquad \qquad a = s \nonumber \\
\displaystyle - \frac{[\gamma] [x_0 - x_s  + \tau + (L+1)\gamma]}{[\tau + (L+1)\gamma] [x_0 - x_s]} \prod_{j=1}^L \frac{[x_s - \mu_j] [x_a - \mu_j + \gamma]}{[x_0 - \mu_j + \gamma]} \nonumber \\ 
\displaystyle \qquad\qquad\qquad\qquad\qquad\qquad\qquad \times \prod_{x \in X_{r, s}^{\bar{0}}} \frac{[x_s - x + \gamma]}{[x_s - x]} \hfill  a = r \nonumber \\
0 \hfill \text{otherwise}
\end{cases} \\
\>
\<
&& \bar{\mathcal{I}}_{n_{l,m}, b} = \nonumber \\
&& \begin{cases}
\displaystyle \frac{[\gamma] [x_m - x_{\bar{0}} + \tau + (L+1)\gamma]}{[\tau + (L+1)\gamma] [x_m - x_{\bar{0}}]} \prod_{j=1}^L [x_{\bar{0}} - \mu_j] \prod_{x \in X_{l, m}^0} \frac{[x_{\bar{0}} - x + \gamma]}{[x_{\bar{0}} - x]} \qquad \qquad  \;\; b = l \nonumber \\
\displaystyle -\frac{[\gamma] [x_l - x_{\bar{0}} + \tau + (L+1)\gamma]}{[\tau + (L+1)\gamma] [x_l - x_{\bar{0}}]} \prod_{j=1}^L \frac{[x_{\bar{0}} - \mu_j] [x_m - \mu_j + \gamma]}{[x_l - \mu_j + \gamma]} \nonumber \\ 
\displaystyle \qquad \qquad \qquad \qquad \qquad \qquad \quad \quad \quad \times  \prod_{x \in X_{l, m}^0} \frac{[x_{\bar{0}} - x + \gamma]}{[x_{\bar{0}} - x]} \hfill  \;\; b = m \nonumber \\
0 \hfill \text{otherwise} 
\end{cases} \\
\>
\<
&& \mathcal{J}_{n_{l,m}, b} = \nonumber \\
&& \begin{cases}
\displaystyle \frac{[\gamma] [x_m - x_0 + \tau + (L+1)\gamma]}{[\tau + (L+1)\gamma] [x_m - x_0]} \prod_{j=1}^L [x_0 - \mu_j] \prod_{x \in X_{l, m}^{\bar{0}}} \frac{[x_0 - x + \gamma]}{[x_0 - x]} \qquad \qquad \qquad \qquad \qquad \;\; b = l \nonumber \\
\displaystyle -\frac{[\gamma] [x_l - x_0 + \tau + (L+1)\gamma]}{[\tau + (L+1)\gamma] [x_l - x_0]} \prod_{j=1}^L \frac{[x_0 - \mu_j] [x_m - \mu_j + \gamma]}{[x_l - \mu_j + \gamma]} \nonumber \\ 
\displaystyle \qquad \qquad \qquad \qquad \qquad \qquad \quad \quad \quad \times  \prod_{x \in X_{l, m}^{\bar{0}}} \frac{[x_0 - x + \gamma]}{[x_0 - x]} \hfill  \;\; b = m \nonumber \\
0 \hfill \text{otherwise} 
\end{cases} \\
\>
\<
&& \mathcal{K}_{n_{l,m}, n_{r,s}} = \nonumber \\
&&\begin{cases}
\displaystyle \frac{[x_l - x_0 + \gamma] [x_l - x_{\bar{0}} + \gamma]}{[x_l - x_0 ] [x_l - x_{\bar{0}}]} \prod_{x \in X_{l,m}} \frac{[x_l - x + \gamma]}{[x_l - x]} \prod_{j=1}^L \frac{[x_l - \mu_j] [x_m - \mu_j + \gamma]}{[x_l - \mu_j + \gamma]} \nonumber \\
\displaystyle  \qquad - \; \frac{[x_m - x_0 + \gamma] [x_m - x_{\bar{0}} + \gamma]}{[x_m - x_0 ] [x_m - x_{\bar{0}}]} \prod_{x \in X_{l,m}} \frac{[x_m - x + \gamma]}{[x_m - x]}  \prod_{j=1}^L [x_m - \mu_j] \qquad \quad l = r , \; m = s \nonumber \\
\displaystyle \frac{[\gamma] [x_m - x_s + \tau + (L+1)\gamma]}{[\tau + (L+1)\gamma] [x_m - x_s]} \prod_{j=1}^L [x_s - \mu_j] \prod_{x \in X_{l, m, s}^{0 , \bar{0}}} \frac{[x_s - x + \gamma]}{[x_s - x]} \hfill l=r , \; m \neq s \nonumber \\
\displaystyle - \frac{[\gamma] [x_l - x_s + \tau + (L+1)\gamma]}{[\tau + (L+1)\gamma] [x_l - x_s]} \prod_{j=1}^L \frac{[x_s - \mu_j] [x_m - \mu_j + \gamma]}{[x_l - \mu_j + \gamma]} \nonumber \\
\displaystyle \qquad \qquad \qquad \;\; \qquad \qquad \qquad \qquad \times \prod_{x \in X_{l, m, s}^{0 , \bar{0}}} \frac{[x_s - x + \gamma]}{[x_s - x]} \hfill l \neq s , \; m = r   \nonumber \\
\displaystyle \frac{[\gamma] [x_m - x_r + \tau + (L+1)\gamma]}{[\tau + (L+1)\gamma] [x_m - x_r]} \prod_{j=1}^L [x_r - \mu_j] \prod_{x \in X_{l,m, r}^{0, \bar{0}}} \frac{[x_r - x + \gamma]}{[x_r - x]} \hfill l=s , \; m \neq r \\
\displaystyle - \frac{[\gamma] [x_l - x_r + \tau + (L+1)\gamma]}{[\tau + (L+1)\gamma] [x_l - x_r]} \prod_{j=1}^L \frac{[x_r - \mu_j][x_m - \mu_j + \gamma]}{[x_l - \mu_j + \gamma]} \nonumber \\
\displaystyle \qquad \qquad \qquad \;\; \qquad \qquad \qquad \qquad \times \prod_{x \in X_{l, m, r}^{0 , \bar{0}}} \frac{[x_r - x + \gamma]}{[x_r - x]} \hfill l \neq r , \; m = s , \;  \nonumber \\
0 \hfill \text{otherwise}
\end{cases} \; . \\
\>

On the other hand, the construction of $\Omega_i$ appearing in representation \eqref{Z0I} also requires explicit expressions for three additional matrices, namely
$\mathcal{F}_i$, $\bar{\mathcal{I}}_i$ and $\bar{\mathcal{F}}_i$. They read as follows:
\<
&& ( \mathcal{F}_i )_{a,b} = \nonumber \\
&& \begin{cases}
\displaystyle \frac{[\gamma] [x_a - x_0 + \tau + (L+1)\gamma]}{[\tau + (L+1)\gamma] [x_a - x_0]} \prod_{j=1}^L \frac{[x_0 - \mu_j] [x_{\bar{0}} - \mu_j + \gamma]}{[x_a - \mu_j + \gamma]} \nonumber \\
\displaystyle \qquad\quad\qquad\qquad\qquad\qquad\qquad\qquad \times \prod_{x \in X_a} \frac{[x_0 - x + \gamma]}{[x_0 - x]} \qquad \qquad \qquad \qquad \quad \;\; i = b \nonumber \\
\displaystyle \frac{[x_a - x_0 + \gamma]}{[x_a - x_0]} \prod_{x \in X_a} \frac{[x_a - x + \gamma]}{[x_a - x]} \prod_{j=1}^L \frac{ [x_a - \mu_j] [x_{\bar{0}} - \mu_j + \gamma]}{ [x_a - \mu_j + \gamma]} \nonumber \\
\displaystyle \qquad  - \; \frac{[x_{\bar{0}} - x_0 + \gamma] }{[x_{\bar{0}} - x_0 ]} \prod_{x \in X_a} \frac{[x_{\bar{0}} - x + \gamma] }{[x_{\bar{0}} - x ]} \prod_{j=1}^L [x_{\bar{0}} - \mu_j] \hfill  \;\; i \neq b , \; a = b \nonumber \\
\displaystyle - \frac{[\gamma] [x_a - x_b + \tau + (L+1)\gamma]}{[\tau + (L+1)\gamma] [x_a - x_b]} \prod_{j=1}^L \frac{[x_b - \mu_j] [x_{\bar{0}} - \mu_j + \gamma]}{[x_a - \mu_j + \gamma]}  \nonumber \\
\displaystyle \;\quad\qquad\qquad\qquad\qquad\qquad\qquad \times \prod_{x \in X_{a,b}^0} \frac{[x_b - x + \gamma]}{[x_b - x]} \hfill  \;\; i \neq b , \; a \neq b
\end{cases} \\
\>
\<
&& (\bar{\mathcal{I}}_i)_{n_{l,m}, b} = \nonumber \\
&& \begin{cases}
0 \hfill  \;\; i = b \nonumber \\
\displaystyle \frac{[\gamma] [x_m - x_{\bar{0}} + \tau + (L+1)\gamma]}{[\tau + (L+1)\gamma] [x_m - x_{\bar{0}}]} \prod_{j=1}^L [x_{\bar{0}} - \mu_j] \prod_{x \in X_{l, m}^0} \frac{[x_{\bar{0}} - x + \gamma]}{[x_{\bar{0}} - x]} \qquad \qquad  \;\; i\neq b , \; b = l \nonumber \\
\displaystyle -\frac{[\gamma] [x_l - x_{\bar{0}} + \tau + (L+1)\gamma]}{[\tau + (L+1)\gamma] [x_l - x_{\bar{0}}]} \prod_{j=1}^L \frac{[x_{\bar{0}} - \mu_j] [x_m - \mu_j + \gamma]}{[x_l - \mu_j + \gamma]} \nonumber \\ 
\displaystyle \qquad \qquad \qquad \qquad \qquad \qquad \quad \quad \quad \times  \prod_{x \in X_{l, m}^0} \frac{[x_{\bar{0}} - x + \gamma]}{[x_{\bar{0}} - x]} \hfill  \;\; i \neq b , \; b = m \nonumber \\
0 \hfill  \;\; i \neq b , \; b \neq l, m
\end{cases} \\
\>
\<
&& (\bar{\mathcal{F}}_i)_{a,b} = \nonumber \\
&&\begin{cases}
\displaystyle - \frac{[\gamma] [x_a - x_{\bar{0}} + \tau + (L+1)\gamma]}{[\tau + (L+1)\gamma] [x_a - x_{\bar{0}}]} \prod_{j=1}^L [x_{\bar{0}} - \mu_j] \prod_{x \in X_a} \frac{[x_{\bar{0}} - x + \gamma]}{[x_{\bar{0}} - x]} \qquad \qquad \qquad \qquad \quad \;\; i = b \nonumber \\
\displaystyle - \frac{[\gamma] [x_0 - x_{\bar{0}} + \tau + (L+1)\gamma]}{[\tau + (L+1)\gamma] [x_0 - x_{\bar{0}}]} \prod_{j=1}^L \frac{[x_{\bar{0}} - \mu_j] [x_a - \mu_j + \gamma]}{[x_0 - \mu_j + \gamma]} \nonumber \\
\displaystyle \;\;\qquad\qquad\qquad\quad\qquad\qquad\qquad\qquad \times \prod_{x \in X_a} \frac{[x_{\bar{0}} - x + \gamma]}{[x_{\bar{0}} - x]} \hfill  \;\; i \neq b, \; a = b \nonumber \\
0 \hfill  \;\; i \neq b, \; a \neq b
\end{cases} \; . \\
\>

Next we consider the additional matrices $\mathcal{G}_i$, $\mathcal{J}_i$ and $\bar{\mathcal{G}}_i$ required for building representation
\eqref{Zb0I} through the matrix $\bar{\Omega}_i$ defined in \eqref{BOIJ}. These elements are given by:
\<
&& ( \mathcal{G}_i )_{a,b} = \nonumber \\
&&\begin{cases}
\displaystyle \frac{[\gamma] [x_a - x_0 + \tau + (L+1)\gamma]}{[\tau + (L+1)\gamma] [x_a - x_0]} \prod_{j=1}^L \frac{[x_0 - \mu_j] [x_{\bar{0}} - \mu_j + \gamma]}{[x_a - \mu_j + \gamma]} \nonumber \\
\displaystyle \qquad\qquad\qquad\qquad\qquad\qquad \;\; \quad \quad \times \prod_{x \in X_a} \frac{[x_0 - x + \gamma]}{[x_0 - x ]} \qquad\qquad\qquad \qquad \qquad  \;\;  i = b \nonumber \\
\displaystyle \frac{[\gamma] [x_{\bar{0}} - x_0 + \tau + (L+1)\gamma]}{[\tau + (L+1)\gamma] [x_{\bar{0}} - x_0]} \prod_{j=1}^L [x_0 - \mu_j] \prod_{x \in X_a} \frac{[x_0 - x + \gamma]}{[x_0 - x]} \hfill  \;\; i \neq b , \; a = b \nonumber \\
0 \hfill  \;\; i \neq b , \; a \neq b
\end{cases} \\
\>
\<
&& ( \mathcal{J}_i )_{n_{l,m}, b} = \nonumber \\
&& \begin{cases}
0 \hfill \;\; i = b \nonumber \\
\displaystyle \frac{[\gamma] [x_m - x_0 + \tau + (L+1)\gamma]}{[\tau + (L+1)\gamma] [x_m - x_0]} \prod_{j=1}^L [x_0 - \mu_j] \prod_{x \in X_{l, m}^{\bar{0}}} \frac{[x_0 - x + \gamma]}{[x_0 - x]} \qquad \qquad  \;\; i \neq b , \; b = l \nonumber \\
\displaystyle -\frac{[\gamma] [x_l - x_0 + \tau + (L+1)\gamma]}{[\tau + (L+1)\gamma] [x_l - x_0]} \prod_{j=1}^L \frac{[x_0 - \mu_j] [x_m - \mu_j + \gamma]}{[x_l - \mu_j + \gamma]} \nonumber \\ 
\displaystyle \qquad \qquad \qquad \qquad \qquad \qquad \quad \quad \quad \times  \prod_{x \in X_{l, m}^{\bar{0}}} \frac{[x_0 - x + \gamma]}{[x_0 - x]} \hfill  \;\; i \neq b, \; b = m \nonumber \\
0 \hfill \;\; i \neq b, \; b \neq l, m
\end{cases} \\
\>
\<
&& (\bar{\mathcal{G}}_i)_{a,b} = \nonumber \\
&&\begin{cases}
\displaystyle - \frac{[\gamma] [x_a - x_{\bar{0}} + \tau + (L+1)\gamma]}{[\tau + (L+1)\gamma] [x_a - x_{\bar{0}}]} \prod_{j=1}^L [x_{\bar{0}} - \mu_j] \prod_{x \in X_a} \frac{[x_{\bar{0}} - x + \gamma]}{[x_{\bar{0}} - x]} \qquad \qquad \qquad \qquad \quad  \;\; i = b \nonumber \\
\displaystyle \frac{[x_0 - x_{\bar{0}} + \gamma]}{[x_0 - x_{\bar{0}}]} \prod_{x \in X_a} \frac{[x_0 - x + \gamma]}{[x_0 - x]} \prod_{j=1}^L \frac{[x_0 - \mu_j] [x_a - \mu_j + \gamma]}{[x_0 - \mu_j + \gamma]} \nonumber \\
\displaystyle \qquad \qquad - \; \frac{[x_a - x_{\bar{0}} + \gamma]}{[x_a - x_{\bar{0}}]} \prod_{x \in X_a} \frac{[x_a - x + \gamma]}{[x_a - x]} \prod_{j=1}^L [x_a - \mu_j] \hfill  \;\; i \neq b, \; a = b \nonumber \\
\displaystyle \frac{[\gamma] [x_a - x_b + \tau + (L+1)\gamma]}{[\tau + (L+1)\gamma] [x_a - x_b]} \prod_{j=1}^L [x_b - \mu_j] \prod_{x \in X_{a,b}^{\bar{0}}}  \frac{[x_b - x + \gamma]}{[x_b - x]} \hfill  \;\; i \neq b, \; a \neq b
\end{cases} \; . \\
\>

Lastly, we also need explicit formulae for the matrices $\mathcal{I}_{ij}$, $\mathcal{K}_{ij}$ and $\bar{\mathcal{J}}_{ij}$. They are required in order to build
matrix $\widetilde{\Omega}_{ij}$ and consequently representation \eqref{Z0b0IJ}. Their explicit expressions are the following:
\<
&& (\mathcal{I}_{ij})_{a , n_{r,s}} = \nonumber \\
&& \begin{cases}
\displaystyle \frac{[\gamma] [x_a - x_0 + \tau + (L+1)\gamma]}{[\tau + (L+1)\gamma] [x_a - x_0]} \prod_{j=1}^L \frac{[x_0 - \mu_j] [x_{\bar{0}} - \mu_j + \gamma]}{[x_a - \mu_j + \gamma]} \nonumber \\
\displaystyle \qquad\qquad\qquad\qquad\qquad\qquad \;\; \quad \quad \times \prod_{x \in X_a} \frac{[x_0 - x + \gamma]}{[x_0 - x ]} \qquad\qquad\qquad \quad \;\;\; \;\;  i = r, \; j=s \nonumber \\
\displaystyle \frac{[\gamma] [x_{\bar{0}} - x_r + \tau + (L+1)\gamma]}{[\tau + (L+1)\gamma] [x_{\bar{0}} - x_r]} \prod_{j=1}^L [x_r - \mu_j] \prod_{x \in X_{r,s}^0} \frac{[x_r - x + \gamma]}{[x_r - x]} \hfill  \;\;  i \neq r, \; j \neq s , \;  a = s \nonumber \\
\displaystyle \frac{[\gamma] [x_{\bar{0}} - x_s + \tau + (L+1)\gamma]}{[\tau + (L+1)\gamma] [x_{\bar{0}} - x_s]} \prod_{j=1}^L [x_s - \mu_j] \prod_{x \in X_{r,s}^0} \frac{[x_s - x + \gamma]}{[x_s - x]} \hfill  \;\;  i \neq r, \; j \neq s , \;  a = r \nonumber \\
0 \hfill  \;\;  i \neq r, \; j \neq s , \; a \neq r, s 
\end{cases} \\
\>
\<
&& (\mathcal{K}_{ij} )_{n_{l,m}, n_{r,s}} = \nonumber \\
&&\begin{cases}
0 \hfill \;\; i=r , \; j=s \nonumber \\
\displaystyle \frac{[x_l - x_0 + \gamma] [x_l - x_{\bar{0}} + \gamma]}{[x_l - x_0 ] [x_l - x_{\bar{0}}]} \prod_{x \in X_{l,m}} \frac{[x_l - x + \gamma]}{[x_l - x]} \prod_{j=1}^L \frac{[x_l - \mu_j] [x_m - \mu_j + \gamma]}{[x_l - \mu_j + \gamma]} \nonumber \\
\displaystyle \qquad - \; \frac{[x_m - x_0 + \gamma] [x_m - x_{\bar{0}} + \gamma]}{[x_m - x_0 ] [x_m - x_{\bar{0}}]} \prod_{x \in X_{l,m}} \frac{[x_m - x + \gamma]}{[x_m - x]}  \prod_{j=1}^L [x_m - \mu_j] \qquad \qquad    \substack{i \neq r , \; j \neq s \\ \; l = r , \; m = s} \nonumber \\
\displaystyle \frac{[\gamma] [x_m - x_s + \tau + (L+1)\gamma]}{[\tau + (L+1)\gamma] [x_m - x_s]} \prod_{j=1}^L [x_s - \mu_j] \prod_{x \in X_{l, m , s}^{0 , \bar{0}}} \frac{[x_s - x + \gamma]}{[x_s - x]} \hfill  \;\; \substack{i \neq r , \;  j \neq s \\ l = r , \; m \neq s} \nonumber \\
\displaystyle - \frac{[\gamma] [x_l - x_s + \tau + (L+1)\gamma]}{[\tau + (L+1)\gamma] [x_l - x_s]} \prod_{j=1}^L \frac{[x_s - \mu_j] [x_m - \mu_j + \gamma]}{[x_l - \mu_j + \gamma]} \nonumber \\
\displaystyle \qquad \qquad \qquad \;\; \qquad \qquad \qquad \qquad \times \prod_{x \in X_{l, m, s}^{0 , \bar{0}}} \frac{[x_s - x + \gamma]}{[x_s - x]} \hfill  \;\; \substack{i \neq r , \;  j \neq s \\ l \neq s , \; m = r } \nonumber \\
\displaystyle \frac{[\gamma] [x_m - x_r + \tau + (L+1)\gamma]}{[\tau + (L+1)\gamma] [x_m - x_r]} \prod_{j=1}^L [x_r - \mu_j] \prod_{x \in X_{l,m, r}^{0, \bar{0}}} \frac{[x_r - x + \gamma]}{[x_r - x]} \hfill  \;\; \substack{i \neq r , \;  j \neq s \\ l = s , \; m \neq r} \\
\displaystyle - \frac{[\gamma] [x_l - x_r + \tau + (L+1)\gamma]}{[\tau + (L+1)\gamma] [x_l - x_r]} \prod_{j=1}^L \frac{[x_r - \mu_j][x_m - \mu_j + \gamma]}{[x_l - \mu_j + \gamma]} \nonumber \\
\displaystyle \qquad \qquad \qquad \;\; \qquad \qquad \qquad \qquad \times \prod_{x \in X_{l, m, r}^{0 , \bar{0}}} \frac{[x_r - x + \gamma]}{[x_r - x]} \hfill  \;\; \substack{i \neq r , \;  j \neq s \\ l \neq r , \; m = s } \nonumber \\
0 \hfill  \;\; \substack{i \neq r , \;  j \neq s \\ l \neq r, s \; , \; m \neq r, s} 
\end{cases} \\
\>
\<
&& (\bar{\mathcal{J}}_{ij})_{a , n_{r,s}} = \nonumber \\
&& \begin{cases}
\displaystyle - \frac{[\gamma] [x_a - x_{\bar{0}} + \tau + (L+1)\gamma]}{[\tau + (L+1)\gamma] [x_a - x_{\bar{0}}]} \prod_{j=1}^L [x_{\bar{0}} - \mu_j] \prod_{x \in X_a} \frac{[x_{\bar{0}} - x + \gamma]}{[x_{\bar{0}} - x]} \qquad \qquad  \;\; i = r, \; j = s \nonumber \\
\displaystyle - \frac{[\gamma] [x_0 - x_r  + \tau + (L+1)\gamma]}{[\tau + (L+1)\gamma] [x_0 - x_r]} \prod_{j=1}^L \frac{[x_r - \mu_j] [x_a - \mu_j + \gamma]}{[x_0 - \mu_j + \gamma]} \nonumber \\
\displaystyle \qquad\qquad\qquad\qquad\qquad\qquad\qquad \times \prod_{x \in X_{r, s}^{\bar{0}}} \frac{[x_r - x + \gamma]}{[x_r - x]} \hfill  \;\; i \neq r, \; j \neq s, \; a = s \nonumber \\
\displaystyle - \frac{[\gamma] [x_0 - x_s  + \tau + (L+1)\gamma]}{[\tau + (L+1)\gamma] [x_0 - x_s]} \prod_{j=1}^L \frac{[x_s - \mu_j] [x_a - \mu_j + \gamma]}{[x_0 - \mu_j + \gamma]} \nonumber \\ 
\displaystyle \qquad\qquad\qquad\qquad\qquad\qquad\qquad \times \prod_{x \in X_{r, s}^{\bar{0}}} \frac{[x_s - x + \gamma]}{[x_s - x]} \hfill  \;\; i \neq r, \; j \neq s, \;  a = r \nonumber \\
0 \hfill  \;\; i \neq r, \; j \neq s , \; a \neq r,s  
\end{cases} \; . \\
\>

\bibliographystyle{alpha}
\bibliography{references}

\end{document}